\documentclass[journal]{IEEEtran}
\ifCLASSINFOpdf
\else
\fi

\usepackage{mathrsfs}
\usepackage{bm}
\usepackage{tabularx}
\usepackage{booktabs}
\usepackage{threeparttable}
\usepackage{graphicx}
\usepackage{epstopdf}
\usepackage{caption}
\usepackage[fleqn]{amsmath}
\usepackage{amsfonts,amssymb}
\usepackage{ulem}
\usepackage{color}
\usepackage{subfigure}
\usepackage[caption=false, font=normalsize, labelfont=sf, textfont=sf]{subfig}

\newtheorem{definition}{Defination}[section]
\newtheorem{lemma}{Lemma}[section]
\newtheorem{theorem}{Theorem}[section]
\newenvironment{proof}{Proof.}{\hfill$\square$\par}
\newtheorem{remark}{Remark}[section]
\newtheorem{asp}{Assumption}[section]
\newtheorem{corollary}{Corollary}[section]

\makeatletter  
\newif\if@restonecol  
\makeatother

\usepackage[linesnumbered,ruled,lined]{algorithm2e}
\usepackage{algpseudocode}

\SetKwFor{For}{for}{do}{endfor}

\begin{document}
%
\title{Decentralized Wireless Federated Learning with Differential Privacy}
\author{Shuzhen Chen,~\IEEEmembership{Member,~IEEE,}
		Dongxiao Yu,~\IEEEmembership{Senior Member,~IEEE,}
		Yifei Zou,~\IEEEmembership{Member,~IEEE,}
		Jiguo Yu,~\IEEEmembership{Fellow,~IEEE,}
		and~Xiuzhen Cheng,~\IEEEmembership{Fellow,~IEEE}
		\thanks{S. Chen, D. Yu, Y. Zou and X. Cheng are with the School of Computer Science and Technology, Shandong University, Qingdao, 266237, P.R. China.\protect\\
			(e-mail: \{szchen\}@mail.sdu.edu.cn, \{dxyu,yfzou, xzcheng\}@sdu.edu.cn).}
			\thanks{J. Yu is with the School of Computer Science and Technology, Qilu University of Technology (Shandong Academy of Science), Jinan 250014, P.R. China, and also with the Shandong Computer Science Center (National Supercomputer Center in Jinan), Jinan 250014, P.R. China. (e-mail: jiguoyu@sina.com).}
		}
\markboth{Journal of \LaTeX\ Class Files,~Vol.~14, No.~8, August~2015}%
{Shell \MakeLowercase{\textit{et al.}}: Bare Demo of IEEEtran.cls for IEEE Journals}
%



\maketitle

\begin{abstract}
This paper studies decentralized federated learning algorithms in wireless IoT networks. The traditional parameter server architecture for federated learning faces some problems such as low fault tolerance, large communication overhead and inaccessibility of private data. To solve these problems, we propose a \uline{D}ecentralized-\uline{W}ireless-\uline{F}ederated-\uline{L}earning algorithm called DWFL. The algorithm works in a system where the workers are organized in a peer-to-peer and server-less manner, and the workers exchange their privacy preserving data with the analog transmission scheme over wireless channels in parallel.
With rigorous analysis, we show that DWFL satisfies $(\epsilon,\delta)$-differential privacy and the privacy budget per worker scales as $\mathcal{O}(\frac{1}{\sqrt{N}})$, in contrast with the constant budget in the orthogonal transmission approach. Furthermore, DWFL converges at the same rate of $\mathcal{O}(\sqrt{\frac{1}{TN}})$ as the best known centralized algorithm with a central parameter server. Extensive experiments demonstrate that our algorithm DWFL also performs well in real settings.

\end{abstract}


%
\IEEEpeerreviewmaketitle

\section{Introduction}\label{sec:introduction}
More and more intelligent devices have emerged in Internet of Things (IoT) systems in recent years, including all kinds of wearable smart devices and intelligent sensory elements. Smart devices play an important role in fields like smart healthcare, smart cities, transportation and automated systems~\cite{22DBLP:journals/access/UllahAAAHJ21}. Meanwhile, a large amount of data has been generated. Faced with the pressure of large amount of data transmission, processing and analysis, federated learning~(FL) becomes popular in distributed machine learning. The workers in FL do not need to send the local raw data set to the central parameter server~(PS)~\cite{1DBLP:conf/aistats/McMahanMRHA17}. Instead, each worker only needs to transmit the locally trained model to PS. In contrast with centralized learning, FL exhibits several special advantages. For example, the centralized approach is inefficient in terms of a large amount of computation, whereas FL can fully make use of the growing computation power of smart devices; The local raw data for each worker is never shared such that each worker' privacy is more possible to be protected. It is thus highly conducive to realizing privacy-enhanced IoT networks.

On the other hand, the development of FL still faces some great challenges. First, the traditional parameter server architecture is fragile, since the entire system is centered at the PS and any failure on the PS will make the whole system broken down. Second, the frequent communication of many wireless applications in IoT systems consumes a lot of bandwidth resource. When a large number of devices participate in the learning process, the bandwidth will become the bottleneck of the learning system especially in large networks~\cite{25DBLP:journals/corr/abs-2103-16055}. Third, the privacy leakage has not been totally addressed and recent studies, e.g.~\cite{11DBLP:conf/sp/ShokriSSS17,12DBLP:journals/popets/HayesMDC19,13DBLP:conf/sp/MelisSCS19}, have shown that there exist attacks that can violate privacy even if only the model parameters or gradients are exchanged between machines. To address these challenges, it is crucial to consider efficient wireless FL with privacy protection.

Most existing works in wireless FL consider improving system scalability~\cite{5DBLP:journals/twc/ZhuWH20,8DBLP:journals/jsac/WangTSLMHC19,9DBLP:journals/twc/AmiriG20}, privacy protection~\cite{10DBLP:journals/jsac/MohamedCT21,20DBLP:conf/isit/SeifTL20} and communication efficiency~\cite{4DBLP:journals/tsp/AmiriG20,6DBLP:journals/twc/YangJSD20,7DBLP:conf/spawc/AmiriG19}. Then natural questions arise. Can wireless FL perform well with privacy protection under Differential Privacy (DP) in decentralized topology? Is there a better tradeoff between the convergence of FL and the privacy protection of workers by utilizing the superposition property of wireless channels?

In this paper, we devise a Decentralized Wireless privacy-preserving Federated Learning (DWFL) algorithm that can be generally and efficiently implemented in wireless IoT networks, to answer the above questions affirmatively. The DWFL algorithm can significantly speed up the training process and avoid the fatal single failure compared with transferring all the data to the PS in the centralized topology, such that any single failure will not strike a fatal blow and break the learning process.

Faced with the frequent communication between workers in the decentralized topology, our algorithm DWFL fully takes advantage of the superposition property of the analog scheme~\cite{23DBLP9441051} to improve the parallelization of communication, such that the number of communication rounds is greatly reduced. Based on the analog scheme, workers transmit their message simultaneously and the receiver will obtain the aggregated result through the over-the-air computation channel directly instead of individual messages from each of the workers. In contrast, the digital scheme requires quantization and channel encoding/decoding, which is always relying on a trusted PS. In addition, the digital scheme of transmitting and reconstructing all information entries one-by-one is an overkill because of the bandwidth limitations posed by practical wireless communications~\cite{25DBLP:journals/corr/abs-2103-16055,10DBLP:journals/jsac/MohamedCT21,4DBLP:journals/tsp/AmiriG20}. Thus, the digital scheme is not applicable to our high-efficiency algorithm. Hence, we here adopt the anolog scheme, where all users can simultaneously use the same time-frequency resources to transmit messages.

We also provide analyzable privacy protection for workers under the requirement of DP~\cite{14DBLP:journals/fttcs/DworkR14}. Our approach allows workers to transmit privacy-preserving data perturbed by random Gaussian noise instead of their raw data. In our analog-based algorithm, it shows that the per-user privacy level behaves at $\mathcal{O}(1/\sqrt{N})$. It means that the privacy budget decays with the number of total workers $N$, which is desirable in the applications of large-scale IoT networks.

The main contributions of this paper are summarized as follows:
\begin{itemize}
	\item We propose DWFL, an innovative, robust and efficient wireless federated learning algorithm, where the problem of single failure, limited bandwidth resource can be solved.
	\item We give the privacy and convergence analysis of DWFL. It is shown that DWFL satisfies $(\epsilon,\delta)$-differential privacy and has great advantages in noise resistance comparing to the orthogonal transmission scheme. Furthermore, DWFL converges at the same rate $\mathcal{O}(\sqrt{\frac{1}{TN}})$ as the centralized algorithm where all workers are connected to a PS. 
	\item We conduct extensive experiments to illustrate the performance of decentralized privacy-preserving wireless FL algorithm on public datasets. It demonstrates that our algorithm DWFL performs well in real settings.
\end{itemize} 

\section{Related Work}\label{sec:rw}

As a conventional approach of distributed learning, federated learning is receiving more and more attention. In the traditional federated learning, unavailability of any central parameter server results in an immediate and complete interruption of the training process. In order to reduce the risk of single point failure, the decentralized federated learning architecture has been proposed~\cite{3DBLP:journals/spm/LiSTS20}. In decentralized federated learning, clients exchange their local model updates directly in a point-to-point manner. However, considering clients' frequent communication with each other in the learning process, exchanging information can incur a lot of communication overhead. 

The superposition nature of analog schemes naturally enables federated training to involve data aggregation from multiple clients over the air. Wireless federated learning is broadly divided into digital and analog solutions depending on the transmission strategy. Some recent studies have focused on the use of wireless channels to reduce communication overhead of federated learning~\cite{5DBLP:journals/twc/ZhuWH20,8DBLP:journals/jsac/WangTSLMHC19,9DBLP:journals/twc/AmiriG20,10DBLP:journals/jsac/MohamedCT21,4DBLP:journals/tsp/AmiriG20,6DBLP:journals/twc/YangJSD20,7DBLP:conf/spawc/AmiriG19}. \cite{9DBLP:journals/twc/AmiriG20,4DBLP:journals/tsp/AmiriG20} have studied the digital scheme for wireless federated learning. In~\cite{4DBLP:journals/tsp/AmiriG20}, by setting the number of top elements to one value before transmission, the gradient vectors are first locally sparse and quantified at the users. In~\cite{9DBLP:journals/twc/AmiriG20}, Amiri et al. modified the digital scheme~\cite{4DBLP:journals/tsp/AmiriG20} so that only users with optimal channel conditions are allowed to transmit. However, digital schemes require recipients to decode the data sent by individuals and then aggregate them, which burdens the communication network seriously. In addition. the privacy cannot be well protected compared to the privacy-enhanced protection by analog schemes.

In analog schemes, all users transmit the scaled data simultaneously over the wireless channel and obtain the aggregated results directly. The non-orthogonal aggregation makes the analog scheme more bandwidth efficient than the digital one. In~\cite{5DBLP:journals/twc/ZhuWH20,6DBLP:journals/twc/YangJSD20}, the wireless aggregation is accomplished by incorporating user scheduling through power control or beamforming to enhance the communication efficiency. In~\cite{8DBLP:journals/jsac/WangTSLMHC19}, Wang et al. optimized the frequency of global aggregation based on the model, data and system dynamics in terms of the convergence rate of wireless federated learning. In~\cite{7DBLP:conf/spawc/AmiriG19}, the gradients are projected to the low dimension, and only users with good channel conditions are allowed to transmit to improve the communication efficiency. In~\cite{10DBLP:journals/jsac/MohamedCT21}, Seif et al. studied the federated learning over a wireless channel based on user sampling. We are interested in decentralized federated learning with the anolog scheme, which is more scalable and more efficient in communication.

There are also recent works focusing on the privacy protection of federated learning. Although clients never share local data in federated learning, exchanging raw gradients can still leak information~\cite{11DBLP:conf/sp/ShokriSSS17,12DBLP:journals/popets/HayesMDC19,13DBLP:conf/sp/MelisSCS19}. Recently there has been a flurry of interest in designing federated learning algorithms with strong privacy safeguards. Differential privacy~\cite{14DBLP:journals/fttcs/DworkR14} is a standard notion for private data aggregation and analysis. Some studies have designed federated learning algorithms satisfying differential privacy~\cite{10DBLP:journals/jsac/MohamedCT21,20DBLP:conf/isit/SeifTL20,152020DP,16geyer2017differentially,17truex2019hybrid,18wei2020federated,19lu2019differentially}. In~\cite{152020DP}, Huang et al, designed the DPAGD-CNN algorithm and the DP-FL framework for unbalanced data. Geyer et al. proposed a novel algorithm for client sided differential privacy to preserve federated optimization in~\cite{16geyer2017differentially}. In~\cite{17truex2019hybrid}, a new algorithm combining DP and SMC was proposed to implement federated learning, which can improve the accuracy of the model, guarantee provable privacy, and prevent collusion threats and extraction attacks. In~\cite{18wei2020federated}, Wei et al. proposed a new framework NbAFL based on the notion of differential privacy. Lu et al. proposed a secure and robust asynchronous federated learning scheme based on differential privacy in vehicular networks in~\cite{19lu2019differentially}. In~\cite{10DBLP:journals/jsac/MohamedCT21,20DBLP:conf/isit/SeifTL20}, Seif et al. studied the anolog scheme for wireless federated learning with local differential privacy. However, all existing works do not consider the decentralized wireless setting of FL.

\section{Preliminaries and Model}\label{sec:model}

In this section we will introduce our system model and problem statement in detail. 

We consider the decentralized federated learning. In this model, the workers are organized in a  peer-to-peer and server-less manner. Specifically, each worker can directly communicate with other workers. In our decentralized federated learning model, each worker serves as a center and communicates with all other centers. In this peer-to-peer setting, it is no longer considered that any central point failure causes the learning process to fail.

We also consider the privacy concern in the learning process. In our model, the workers are not allowed to transmit local raw data directly. Instead, the information exchanged by the workers only contains perturbed parameters of the learning model. We divide the learning process into several synchronous rounds. The total number of iterations is $T$. In each round $t$, the worker accomplishes local computation and transmits its information over a wireless channel, modeled by a Gaussian multiple access channel~(MAC). 
In our peer-to-peer setting, each worker receives signal from all other $(N-1)$ workers. 
The input-output relationship of Gaussian MAC at round $t$ is
\begin{equation}\label{equ:mac}
	v_i^{(t)} = \sum_{k\neq i} h_{k}\tilde{x}_k^{(t)} + m_{i}^{(t)},
\end{equation}
where $\tilde{x}_k^{(t)} \in \mathbb{R}^{d}$ is the signal~(\textit{i.e.} perturbed local parameter) transmitted by worker $k$ at round $t$. And $v_i^{(t)}$ is the output of Gaussian MAC at the side of worker $i$. Here $h_{k} = e^{\varrho\theta_{k}}|h_k|$ is a complex valued time-invariant channel coefficient between $k$-th worker and $i$-th worker. $\varrho$ is a constant. For simplicity, we assume this coefficient $h_{k}$ is decided only by sender $k$. Besides the aggregation part $\sum_{k\neq i} h_{k}\tilde{x}_k^{(t)}$, $m_{i}^{(t)} \in \mathbb{R}^d$ is a independent additive zero-mean unit-variance Gaussian noise, caused by inherent 
channel noise. The channel noise is related to the receiving worker $i$. The transmission for each worker $k$ in MAC is constrained by the maximum power of $P_k$.

During the learning process, we consider the following decentralized optimization:
\begin{equation*}
	\min_{x\in \mathbb{R}^{d}} f(x) = \frac{1}{n}\sum_{i=1}^{n}\mathbb{E}_{\xi \sim D_{i}} F_{i}(x;\xi),
\end{equation*}
where $D_{i}$ is the local data set for worker $i$. 
$F_{i}(x,\xi)$ denotes the value of loss function for worker $i$ 
given parameter $x$ and data sample $\xi$. 

Let $f_{i}(x) = \mathbb{E}_{\xi \sim D_{i}}F_{i}(x;\xi)$. Particularly, we try to find an upper bound of $\frac{1}{T}\sum_{t=0}^{T-1}\mathbb{E} \|\nabla f( \overline{x}_{t-\frac{1}{2}})\|$ to analyze the convergence rate of our algorithm.

The private guarantee is of vital importance if the participants of decentralized learning do not admit their 
local training data to be shared. Although each worker communicates with its neighbors by transmitting parameters 
instead of sending raw data, the risk of leaking information still exists~\cite{wang2019beyond}.

Differential privacy~\cite{dwork2014algorithmic} is a mechanism to avoid leaking any private information of single worker by adjusting the feedback of query functions, no matter what auxiliary information the malicious adversary worker has. In the context of our decentralized learning, the exchanging process of parameters is a sequence of queries. The differential privacy in our setup is formally defined as follows:

\begin{definition}\label{definition:dp}
	A randomized query $\mathcal{M}$ on training set with domain $\mathcal{D}$ and range $\mathcal{R}$ satisfies 
	$(\epsilon,\delta)$-differential privacy if for any two adjacent inputs $d,d' \in \mathcal{D}$ and for any subset of 
	outputs $\mathcal{S} \subset \mathcal{R}$ it hols that
	\begin{equation*}
		Pr(\mathcal{M}(d)\in\mathcal{S}) \leq e^{\epsilon}Pr(\mathcal{M}(d')\in \mathcal{S})+\delta,
	\end{equation*}
	where the privacy budget $\epsilon$ denotes the privacy lower bound to measure a randomized query $\mathcal{M}$. And $\delta$ denotes the probability of breaking this lower bound.
\end{definition}

If $\delta = 0$, the above  $(\epsilon,\delta)$-differential privacy is called $\epsilon$-differential privacy. The main purpose is to guarantee $\epsilon$-differential privacy in the decentralized federated learning for each worker while reducing the influence of differential privacy mechanism on convergence rate as small as possible.

A notable fact is that if a query function $\mathcal{M}$ is more sensitive to local data, it is harder for $\mathcal{M}$ to keep $\epsilon$-differential privacy when changing $d\in \mathcal{D}$. Hence, the sensitivity of the query function is commonly used in analyzing the differential privacy of a mechanism, which is defined as below.

\begin{definition}\label{definition:sensitivity}
	For $f:\mathcal{D} \rightarrow \mathbb{R}^N$, the $L2$-sensitivity of $f$ is defined as
	\begin{equation*}
		\Delta_2 f = \max_{d_1,d_2} \|f(d_1) - f(d_2) \|,
	\end{equation*}
	where all $d_1, d_2$ differs in at most one element.
\end{definition}

	The notations used in this paper are summarized in Table \ref{ta:1}.
	\begin{table}[]
		\caption{Frequent Notations and Descriptions}
		\begin{tabular}[t]{p{2cm} p{6cm}}
			\hline
			Notations & Descriptions \\
			\hline
			$N$ & The number of total workers\\
			$d$ & The dimension of local parameter\\
			$\gamma$ & The step size\\
			$\eta$ & The averaging rate\\
			$\epsilon$ & The budget of differential privacy\\
			$\delta$ & The slack variable of differential privacy\\
			$D_{i}$ & The local data set for worker $i$\\
			$T$ & The total number of iterations\\
			$t$ & The current iteration round\\
			$F_{i}(x,\xi)$ & The loss function for worker $i$ \\
			$f_{i}(x)$ & $\mathbb{E}_{\xi \sim D_{i}}F_{i}(x;\xi)$\\
			$v_i^{(t)}$ & The input of Gaussian MAC at worker $i$\\
			$h_{k}$ &  $e^{\varrho\theta_{k}}|h_k|$\\
			$\tilde{x}_k^{(t)}$ & The perturbed local parameter\\
			$m_{i}^{(t)}$ & The additive channel Gaussian noise\\
			$P_k$ & The maximum constrained power for worker $k$\\
			$\mathcal{G}_i^{(t)}$ & The random Gaussian noise at worker $i$\\
			$g_{i}^{(t)}$ & The gradient of worker $i$ in round $t$\\
			$g_{\max}^{(t)}$ &  $\max_{k} \|g_k^{(t)}\|_2$\\
			$\nabla$$F(\cdot)$ & The gradient of a function $F$\\
			$\left(\frac{1}{N}\right)_N$ & All $\frac{1}{N}$ square matrix with size of $N \times N$\\
			$(1)_{N}$ & All $1$ square matrix with size of $N \times N$\\
			$I$ & An identity matrix\\
			$X_{t}$ & $[x_{1}^{(t)}, x_{2}^{(t)}, \cdots, x_{N}^{(t)}] \in \mathbb{R}^{d\times N}$\\
			$G_{t}$ & $[g_{1}^{(t)}, g_{2}^{(t)}, \cdots, g_{N}^{(t)}]\in \mathbb{R}^{d \times N}$\\
			$\overline{x}_{t}$ & $X_{t}\textbf{1}$\\
			$\overline{G}_{t}$ &$G_{t}\textbf{1}$\\
			$\nabla \overline{f}(X_{t})$ & $\frac{1}{N}\sum_{i=1}^{N}\nabla f_{i}(x_i^{(t)})$\\
			\hline
		\end{tabular}
		\label{ta:1}
	\end{table}

\section{Algorithm and Analysis}\label{sec:algorithm}

In this section we present our decentralized learning strategy over Gaussian MAC. 

We propose a \textbf{D}ecentralized \textbf{W}ireless \textbf{F}ederated \textbf{Learning}~(DWFL) algorithm. The pseudo-code of the algorithm is given in Algorithm~\ref{Algo1}. It mainly contains four steps: \textit{Computing gradient}, \textit{Generating signal}, \textit{Parameter exchange} and \textit{Parameter update}.

We then make a detail analysis on 
our strategy to show its benefit in privacy protection. Furthermore, we also analyze the trade-off between the convergence rate of learning algorithm and privacy leakage.

The whole algorithm is divided into $T$ rounds for a given $T$. At the beginning of each round $t$, we use $x_{i}^{(t-\frac{1}{2})} \in \mathbb{R}$ to denote the local parameter of each worker $i$. Before any communication, works stochastically compute local gradient $g_{i}^{(t)}$ and update their local parameters to obtain $x_{i}^{(t)}$:
\begin{equation}\nonumber
	x_{i}^{(t)} = x_{i}^{(t-\frac{1}{2})} - \gamma g_i^{(t)},
\end{equation}
where $\gamma$ is the step size.

Then, each worker exchanges its local model parameter $\tilde{x}_k^{(t)} \in \mathbb{R}^d$ perturbed by a random Gaussian noise $\mathcal{G}_i^{(t)}$. Specifically, the signal transmitted by worker $i$ in round $t$ is 
\begin{equation}\label{equ:input}
	\tilde{x}_{i}^{(t)} = e^{-\varrho\theta_{i}}\left({\sqrt{\alpha_iP_i}}x_{i}^{(t)} + \sqrt{\beta_{i}P_i}\mathcal{G}_{i}^{(t)}\right).
\end{equation}
Here $\alpha_{i}$ is the fraction of power dedicated to the local parameter $X_{i}^{(t)}$  of worker $i$, while $\beta_{i}$ is the fraction of power dedicated to the random Gaussian noise vector $\mathcal{G}_i^{(t)}$. Each value of $\mathcal{G}_t^{(i)}$ is independently drawn from Gaussian distribution $\mathcal{N}(0,\sigma^2)$. The transmission for each worker $i$ in MAC is constrained by the maximum power of $P_i$. The coefficients of signal $\tilde{x}_{i}^{(t)}$ have to satisfy $(\sqrt{\alpha_i P_i})^2 + (\sqrt{\beta_{i}P_i})^2 \leq P_i$. Therefore $\alpha_i + \beta_i \leq 1$, where $\alpha_i \geq 0$ and $\beta_i \geq 0$. Recall that in Eqt. (\ref{equ:mac}) every input from worker $i$ is multiplied by a complex number $h_i = e^{\varrho\theta_i}|h_i|$. And notice that the input is multiplied by $e^{-j\theta_{i}}$ to obtain $|h_i|$.
It can guarantee that the received channel coefficient is non-negative.

From Eqt. (\ref{equ:mac}) and (\ref{equ:input}), the received signal at worker $i$ in round $t$ can be written as:
\begin{equation*}
	\begin{aligned}
		v_i^{(t)} &= \sum_{k\neq i} h_{k}\tilde{x}_k^{(t)} + m_{i}^{(t)}\\
		&=\sum_{k\neq i}|h_k|\left({\sqrt{\alpha_kP_k}}x_{k}^{(t)} + \sqrt{\beta_{k}P_k}\mathcal{G}_{k}^{(t)}\right) + m_{i}^{(t)}\\
		&=\sum_{k\neq i}|h_k|{\sqrt{\alpha_kP_k}}x_{k}^{(t)} + \sum_{k \neq i}|h_k|\sqrt{\beta_{k}P_k}\mathcal{G}_{k}^{(t)} + m_{i}^{(t)}.
	\end{aligned}
\end{equation*}
Here, $m_{i}^{(t)}$ is the channel noise, whose value is independently drawn from $\mathcal{N}(0,\sigma_m^2)$. In order to aggregate the local parameters over-the-air and make sure that 
worker $i$ receives an unbiased estimate of the average local parameter~(with a constant coefficient) of all other workers from the network, every worker $i$ should carefully choose their own input parameter $\alpha_i$. In particular, worker $i$ should choose $\alpha_i$ to guarantee that
\begin{equation}\label{equ:alpha}
	|h_i|\sqrt{\alpha_iP_i}=c , \forall i \in [N],
\end{equation}
where $c$ is a constant. So that all received local parameters can be aligned.

Thus, we have $\alpha_i = \displaystyle\frac{c^2}{|h_i|^2P_i}$. Since $\forall i, \alpha_i \leq 1$, the constant $c$ 
can be upper bounded by
\begin{equation}\label{equ:c}
	c = \min_j\sqrt{|h_j|^2P_j}.
\end{equation}
By Eqt. (\ref{equ:alpha}) and (\ref{equ:c}), we obtain
\begin{equation*}
	\alpha_i = \frac{\min_{j}|h_j|^2P_j}{|h_i|^2P_i}.
\end{equation*}
It is obvious that the alignment of local parameters is significantly related to the worst effective SNR, \textit{i.e.} $\min_j|h_j|^2P_j$. The constant $c$ can be determined by communicating with each other once at the beginning. By the above aligning scheme, the receiving signal of worker $i$ can be rewritten as:
\begin{equation}\label{equ:for v by c}
	v_i^{(t)} = c\sum_{k\neq i} x_{k}^{(t)} + \sum_{k \neq i}|h_k|\sqrt{\beta_{k}P_k}\mathcal{G}_{k}^{(t)} + m_{i}^{(t)}.
\end{equation}
Then Eqt.~(\ref{equ:for v by c}) can be written as:

\begin{align}\label{euq:v c}
	\frac{v_i^{(t)}}{c} &= \sum_{k\neq i} x_{k}^{(t)} + \frac{\sum_{k \neq i}|h_k|\sqrt{\beta_{k}P_k}\mathcal{G}_{k}^{(t)}}{c} + \frac{m_{i}^{(t)}}{c}  \nonumber \\
	&= \sum_{k\neq i}\big[x_{k}^{(t)}+\frac{|h_k|\sqrt{\beta_{k}P_k}\mathcal{G}_{k}^{(t)}}{c}+\frac{m_{i}^{(t)}}{(N-1)c}\big] \nonumber\\
	&= \sum_{k \neq i}\big(x_{k}^{(t)}+\Phi_{k}^{(t,i)}\big),
\end{align}
where $\Phi_{k}^{(t,i)}$ is the $k$-st column of the matrix
	\begin{equation*}
	\begin{aligned}
\Phi^{(t,i)} =& [\frac{|h_1|\sqrt{\beta_{1}P_1}\mathcal{G}_{t}^{(1)}}{c}+\frac{m_i^{(t)}}{c(N-1)}, \cdots,\\ &\frac{\sqrt{|h_N|\beta_{N}P_N}\mathcal{G}_{t}^{(N)}}{c}+\frac{m_i^{(t)}}{c(N-1)}]. 
	\end{aligned}
	\end{equation*}
	After the communication step, each worker $i$ will receive a collection of $v_i^{(t)}$s ($i\neq j$). Then the workers will perform another local update using the following rule:

	\begin{equation}\label{euq:iteration}
		\begin{aligned}
			x_{i}^{(t+\frac{1}{2})} = x_{i}^{(t)} + \frac{\eta}{c} \left(\frac{v_i^{(t)}}{N-1}-c\left({x}_{i}^{(t)}+\Phi_i^{(t,i)}\right)\right),
		\end{aligned}
	\end{equation}
where $\eta$ is the averaging rate. 

We define a weight matrix $W = \frac{(1)_{N}- I}{N-1}$, where $(1)_{N}$ is a $N\times N$ matrix whose elements are all $1$'s and $I$ is an identity matrix. By Eqt.~(\ref{euq:v c}) and (\ref{euq:iteration}), we have
\begin{equation*}
		\begin{aligned}
			x_{i}^{(t+\frac{1}{2})}
			&=x_{i}^{(t-\frac{1}{2})}-\gamma g_{i}^{(t)}-\eta\left(x_{i}^{(t)}+\Phi_i^{(t,i)}\right)  \\	
			&~~~~+\eta\sum_{k=1}^{N} \left[ W_{ik}\times \left(x_{k}^{(t)}+\Phi_{k}^{(t,i)}\right) \right]\\
			&=\eta\sum_{k=1}^{N}\left( W_{ik}-I_{ik}\right)\left(x_{k}^{(t-\frac{1}{2})}-\gamma g_{i}^{(t)}+\Phi_{k}^{(t,i)}\right)\\
			&~~~~+x_{i}^{(t-\frac{1}{2})}-\gamma g_{i}^{(t)}\\
			&~~~~  \\
		\end{aligned}
	\end{equation*} 

We denote $\Psi =  (1-\eta)I + \eta W$, $X_{t} = [x_{1}^{(t)}, x_{2}^{(t)}, \cdots, x_{N}^{(t)}] \in \mathbb{R}^{d\times N}$, $G_{t} = [g_{1}^{(t)}, g_{2}^{(t)}, \cdots, g_{N}^{(t)}]\in \mathbb{R}^{d \times N}$. Then from a global view, the update rule of all $N$ workers can be represented as a matrix $X_{(t+\frac{1}{2})}$ as follow:

\begin{align}\label{equ:X}
	&~~~~X_{(t+\frac{1}{2})}\nonumber \\
	&= X_{t-\frac{1}{2}}-\gamma G_{t}+\eta \left(X_{t-\frac{1}{2}}-\gamma G_{t} + \Phi_{t}\right)\left(W-I\right)\nonumber \\
	&=\left(X_{t-\frac{1}{2}}-\gamma G_{t}\right)\left(I+\left(W-I\right)\eta\right)+\eta\Phi_{t}\left(W-I\right)\nonumber \\
	&=\left(X_{t-\frac{1}{2}}-\gamma G_{t}\right)\Psi + \Phi_{t}\left(\Psi-I\right).
\end{align}

Define a $N$ dimensional vector $A=\left[\frac{1}{N},\frac{1}{N},\cdots,\frac{1}{N}\right]^\top$. Multiplying both sides of Eqt.~(\ref{equ:X}) by $A$, it can be obtained that
\begin{equation}
	\begin{aligned}
		\overline{x}_{(t+\frac{1}{2})}=\overline{x}_{(t-\frac{1}{2})}-\gamma\overline{G}_{(t)}.
	\end{aligned}
\end{equation}

Here, $\overline{G}_{(t)}$ is exactly the average value of the true gradients of each worker's local parameter (i.e. $x_i^{(t-\frac{1}{2})}$). 

Hence, the parameter update approach in our algorithm is in fact a gradient descent manner performed in a decentralized system from a global view. And we will show that our algorithm can converge quickly while preserving the privacy.
\begin{algorithm}
	\normalem
	\caption{DWFL: Decentralized Wireless Federated Learning}
	\label{Algo1}  
	\LinesNumbered   
	\textbf{Initialize:} $x_{1:N}^{(-\frac{1}{2})}=0_{1:N}$, learning rate $\gamma$, averaging rate $\eta$, number of total iterations $T$ and variance of noise $\sigma^2$.\newline
	\For{$t=0$ to $T-1$}{
		\ForEach{worker $i$ = $1$ to $N$}{
			\tcc{Computing gradient}
			Randomly sample $\xi_{i}^{(t)}$ and compute local stochastic gradient $g_{i}^{(t)}:=\nabla$$F_{i}(x_{i}^{(t-\frac{1}{2})},\xi_{i}^{(t)})$\\
			\tcc{Generating signal}
			Generate local signal with stochastic gradient:\\
			$x_{i}^{(t)} = x_{i}^{(t-\frac{1}{2})} - \gamma g_i^{(t)}$\\
			Randomly generate Gaussian noise $\mathcal{G}_i^{(t)} \in \mathbb{R}^d$ from $\mathcal{N}(0,\sigma^2)$ and add the noise to signal:\\
			$\tilde{x}_{i}^{(t)} = e^{-j\theta_{i}}\left({\sqrt{\alpha_iP_i}}x_{i}^{(t)} + \sqrt{\beta_{i}P_i}\mathcal{G}_{i}^{(t)}\right)$\\
			\tcc{Parameter exchange}
			Broadcast $\tilde{x}_i^{(t)}$ to the network\\
			\tcc{Parameter update}
			Update local parameters $x_{i}^{(t+\frac{1}{2})} = x_{i}^{(t)} + \frac{\eta}{c} \left(\frac{v_i^{(t)}}{N-1}-c\left({x}_{i}^{(t)}+\Phi_i^{(t)}\right)\right)$.
		}
	}
	\KwOut{\textbf{x}}
\end{algorithm}

The overall scheme is summarized in Algorithm \ref{Algo1}. More specifically, each worker $i$ maintains two major local variables: the model parameter $x_{i}^{(t)}$ and the privacy preserving variable $\tilde{x}_i^{(t)}$. The parameter $x_i^{(-\frac{1}{2})}$ for each worker is initialized to $0$. In each round $t$, for each worker $i$, Algorithm \ref{Algo1} mainly contains the following four steps:
\begin{itemize}
	\item \textbf{Computing gradient:} worker $i$ computes the local stochastic gradient $g_{i}^{(t)} = \nabla$$F_{i}(x_{i}^{(t)},\xi_{i}^{(t)})$, where $\xi_{i}^{(t)}$ is randomly sampled from $D_i$.
	\item \textbf{Generating signal:} worker $i$ uses $g_{i}^{(t)}$ to update the local parameter $x_{i}^{(t-\frac{1}{2})}$. Then worker $i$ adds random noise to the parameter $x_{i}^{(t)}$ updated last step to obtain the perturbed parameter $\tilde{x}_i^{(t)}$.
	\item \textbf{Parameter exchange:} worker $i$ broadcasts $\tilde{x}_i^{(t)}$ to the wireless network.
	\item \textbf{Parameter update:} worker $i$ updates local parameter with the rule:\\ $x_{i}^{(t+\frac{1}{2})} = x_{i}^{(t)} + \frac{\eta}{c} \left(\frac{v_i^{(t)}}{N-1}-c\left({x}_{i}^{(t)}+\Phi_i^{(t)}\right)\right)$.
\end{itemize}

The learning process of each worker is demonstrated in Fig. \ref{fig: each worker}. 

\begin{figure*}
	\centering
	\includegraphics[width=\linewidth]{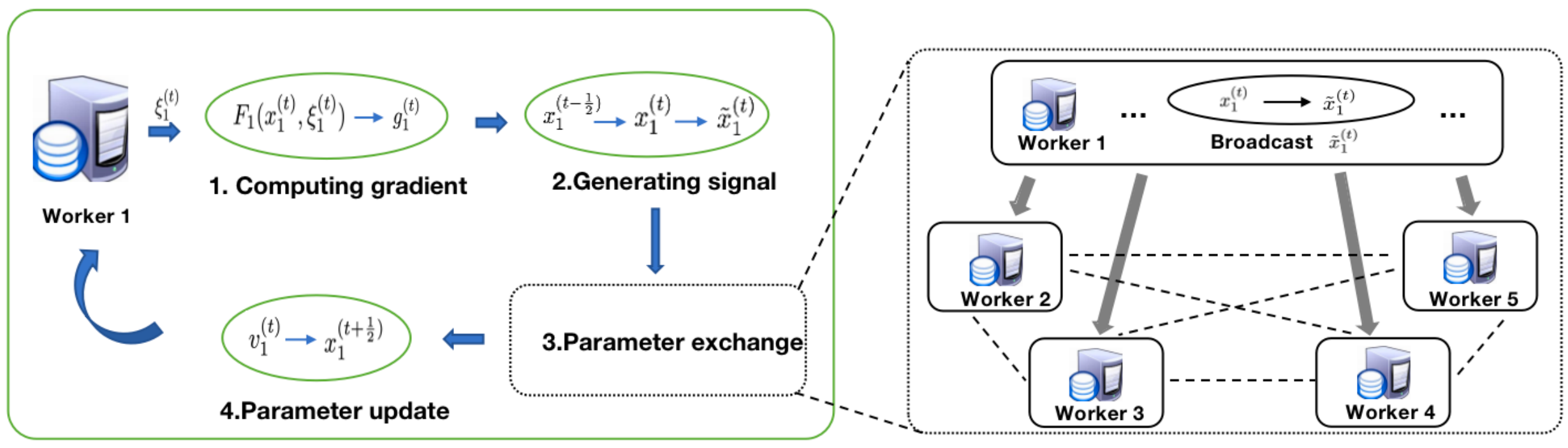}
	\caption{The learning process for each worker.}
	\label{fig: each worker}
\end{figure*}

\subsection{Analysis on Privacy Leakage}

In this part we provide detailed analysis on privacy budget of our wireless decentralized federated learning algorithm. Recall that each parameter are perturbed by random Gaussian noise, and the channel noise is also drawn from Gaussian distribution. Therefore, we introduce the following lemma for analyzing differential privacy with Gaussian mechanism, which can be found in \cite{dwork2014algorithmic}.

\begin{lemma} \label{lemma:Gaussian}
	Let $\epsilon \in (0,1)$ be arbitrary. For constant $a^{2}>2\ln(1.25/\delta)$, the Gaussian Mechanism with parameter $\sigma \geq a\Delta_{2}f/\epsilon$ is $(\epsilon,\delta)$-differentially private, where $f$ is a query function that performs on the dataset.
\end{lemma}
Here, $\Delta_2f$ is the $L2$-sensitivity defined in Definition \ref{definition:sensitivity}. In our setting, each worker receives information of other workers via a wireless channel that automatically aggregates perturbed parameters. In round $t$, The final perturbed aggregating parameter received by $i$ is
\begin{equation}\label{equ:v}
	\begin{aligned}
		v_i^{(t)} &=\sum_{k\neq i}|h_k|{\sqrt{\alpha_kP_k}}x_{k}^{(t)} + \sum_{k \neq i}|h_k|\sqrt{\beta_{k}P_k}\mathcal{G}_{k}^{(t)} + m_{i}^{(t)}\\
		&=c\sum_{k\neq i}x_{k}^{(t)} + \sum_{k \neq i}|h_k|\sqrt{\beta_{k}P_k}\mathcal{G}_{k}^{(t)} + m_{i}^{(t)}.
	\end{aligned}
\end{equation}
In Eqt. (\ref{equ:v}), the first term $\sum_{k\neq i}|h_k|{\sqrt{\alpha_kP_k}}x_{k}^{(t)}$ is the correct aggregation of parameters, which can be considered as the output of query function $f$ in Lemma \ref{lemma:Gaussian}. The second term $\sum_{k \neq i}\sqrt{\beta_{k}P_k}\mathcal{G}_{k}^{(t)} + m_{i}^{(t)}$ is the Gaussian noise with variance $\sigma_s^2 = \sum_{k \neq i}|h_k|^2{\beta_{k}P_k}\sigma^2 +\sigma_m^2$.

We use $g_{\max}^{(t)}$ to bound the gradient of each node in round $t$, that is, $g_{\max}^{(t)} = \max_{k} \|g_k^{(t)}\|_2$. In practice, this constraint can easily be satisfied by clipped gradient. The following theorem describes the privacy budget of each peer to peer channel in Algorithm \ref{Algo1}. 

\begin{theorem}\label{theorem:privacy}
	For each worker $i$, \textbf{Algorithm~\ref{Algo1}} guarantees $(\epsilon_i, \delta)$-differential privacy for any other workers in $i$'s process of collecting information in a round $t$, where
	\begin{equation}\label{equ:epsilon}
		\epsilon_i = \frac{2\gamma g_{\max}^{(t)}\sqrt{\min_j|h_j|^2P_j}}{\sqrt{\sum_{k \neq i}|h_k|^2{\beta_{k}P_k}\sigma^2 + \sigma_m^2}}\sqrt{2\ln \frac{1.25}{\delta}}.
	\end{equation}
\end{theorem}

\begin{proof}
	We first bound the $L2$-sensitivity of $y_i^{(t)} = \sum_{k\neq i}|h_k|{\sqrt{\alpha_kP_k}}x_{k}^{(t)} = c\sum_{k\neq i}x_{k}^{(t)}$. For any other worker $j\neq i$, consider any two different local datasets $D_j$ and $D_j'$. When fixing the datasets of other $(N-1)$ workers, the local sensitivity of $y_i^{(t)}$ can be bounded as:
	\begin{equation*}
		\begin{aligned}
			\Delta_{i,j}^{(t)} &= \max_{D_j,D_j'}\|y_i^{(t)}(D_j) - y_i^{(t)}(D_j')\|_2 \\&= \max_{D_j,D_j'}\left\|c\gamma\left(g_j^{(t)}(D_j) - g_j^{(t)}(D_j')\right)\right\|_2\\
			&\leq c\gamma \max_{D_j,D_j'} \left(\|g_j^{(t)}(D_j)\|_2 + \|g_j^{(t)}(D_j)\|_2\right)\\
			&\leq 2c\gamma g_{\max}^{(t)} = 2\gamma g_{\max}^{(t)}\sqrt{\min_j|h_j|^2P_j}.
		\end{aligned}
	\end{equation*}
	
	In the final inequality, we use the fact that $g_{\max}^{(t)} = \max_{k} \|g_k^{(t)}\|_2$. And in the final equality, we follow from Eqt.~(\ref{equ:c}). Then by applying Lemma \ref{lemma:Gaussian}, Eqt.~(\ref{equ:epsilon}) can be obtained. 
\end{proof}

In order to fully demonstrate the advantages of our strategy over traditional strategies that use the wired channel, we further upper bound $\epsilon_i$.
\begin{remark}
	We can further bound the achievable $\epsilon_i$ in Theorem~\ref{theorem:privacy} as follows:
	\begin{equation*}
		\begin{aligned}
			\epsilon_i &= \frac{2\gamma g_{\max}^{(t)}\sqrt{\min_j|h_j|^2P_j}}{\sqrt{\sum_{k \neq i}|h_k|^2{\beta_{k}P_k}\sigma^2 + \sigma_m^2}}\sqrt{2\ln \frac{1.25}{\delta}}\\
			&\leq  \frac{1}{\sqrt{N-1}} \times \frac{2\gamma g_{\max}^{(t)}\sqrt{\min_j|h_j|^2P_j}}{\sqrt{\min_{k\neq i}|h_k|^2{\beta_{k}P_k}\sigma^2 + \sigma_m^2}}\sqrt{2\ln \frac{1.25}{\delta}}.
		\end{aligned}
	\end{equation*}
\end{remark}

It can be shown that the per-user privacy level behaves like $\mathcal{O}(1/\sqrt{N})$, which means that the privacy budget decays with the number of total workers $N$. Meanwhile, the privacy budget of wired peer to peer channel from $j$ to $i$ with orthogonal transmission can be shown to be:
\begin{equation*}
	\begin{aligned}
		\epsilon_{j\rightarrow i} = \frac{2\gamma g_{\max}^{(t)}\sqrt{|h_j|^2P_j}}{\sqrt{|h_j|^2{\beta_{j}P_j}\sigma^2 + \sigma_m^2}}\sqrt{2\ln \frac{1.25}{\delta}},\\
	\end{aligned}
\end{equation*}
which does not decay with $N$.

\subsection{Analysis on Convergence Rate}

In this part we focus on analyzing the convergence rate of our algorithm. Particularly, we try to find an upper bound of $\frac{1}{T}\sum_{t=0}^{T-1}\mathbb{E} \|\nabla f( \overline{x}_{t-\frac{1}{2}})\|$, which is the average global gradient. We show that the algorithm DWFL attains a convergence rate of $\mathcal{O}(\sqrt{\frac{1}{TN}})$ while keeping the merit of noise resistance.

First, we introduce some assumptions that are commonly used in the analysis of distributed learning.

\begin{asp}\label{assumptions}
	\textit{These are assumptions that we will use in our following analysis:}\\
	\textbf{(1) Lipschitzian Condition:} \textit{All function $f_i(\cdot)$'s are with $L$-Lipschitzian gradients. That is to say:}
	$$\| \nabla f_{i}(x) - \nabla f_{i}(y)\| \leq L\|x - y\|.$$
	\textbf{(2) Bounded variance:} \textit{The variance of stochastic gradient is bounded as follows.}
	$$\mathbb{E}_{\xi\sim D_{i}}\|\nabla F_{i}(x;\xi) - \nabla f_{i}(x)\|^{2} \leq \sigma_{f}^{2}$$
	$$\frac{1}{n}\sum_{i=1}^{n}\|\nabla f_{i}(x) - \nabla f(x)\|^{2} \leq \zeta^{2} ,\quad \forall i,\forall x$$
\end{asp}

Next we show the detailed mathematical analysis.

\begin{lemma}\label{le:vector operations}
	Under \textbf{Assumption~\ref{assumptions}}, we have:
	$$f(x) - f(y) \leq \nabla f(y)^\top(x-y)+\frac{L}{2}\|x-y\|^2.$$
\end{lemma}
	\begin{proof}
		For vectors $x$ and $y$, we have
		\begin{equation*}
			\begin{aligned}
				f(x) &= f(y) + \int_{0}^{1}\langle\nabla f(y)+\psi(x-y),(x-y)\rangle d \psi \\
				&=f(y) + \langle \nabla f(y),x-y \rangle\\
				& +\int_{0}^{1} \langle \nabla f(y+\psi(x-y))-\nabla f(y),x-y \rangle d\psi
			\end{aligned}
		\end{equation*}
		Therefore,
		\begin{equation*}
			\begin{aligned}
				&~~~~~f(x) - f(y)-\langle \nabla f(y),x-y \rangle  \\
				& =\int_{0}^{1} \langle \nabla f(y+\psi(x-y))-\nabla f(y),x-y \rangle d\psi\\
				&\leq \bigg | \int_{0}^{1} \langle \nabla f(y+\psi(x-y))-\nabla f(y),x-y \rangle d\psi \bigg|\\
				&\leq  \int_{0}^{1} \bigg | \langle \nabla f(y+\psi(x-y))-\nabla f(y),x-y \rangle \bigg|d\psi \\
				&\leq  \int_{0}^{1} \|  \nabla f(y+\psi(x-y))-\nabla f(y)\| \cdot \|x-y \| d\psi \\
			\end{aligned}
		\end{equation*}
		According to \textbf{Assumption~\ref{assumptions}}, the \textbf{Lipschitzian Condition} holds. Then it can be obtained that
		\begin{equation*}
			\begin{aligned}
				&~~~~~f(x) - f(y)-\langle \nabla f(y),x-y \rangle  \\
				&\leq \int_{0}^{1} \psi L \|x-y\|^2 d\psi\\
				&=\frac{L}{2}\|x-y\|^2
			\end{aligned}
		\end{equation*}
		Thus,
		$$f(x) - f(y) \leq \nabla f(y)^\top(x-y)+\frac{L}{2}\|x-y\|^2.$$.
	\end{proof}
	{The following {Lemma \ref{le:matrix1}} and {Lemma \ref{le:matrix2}} were proved in~\cite{21DBLP:journals/corr/abs-1907-07346}.}

\begin{lemma}\label{le:matrix1} For any matrix sequence ${M_{t}}$, and positive integer constant $m\in 
	\{1, 2,...\}$, we have
	\begin{align*}
		~~~~&\sum_{t=0}^{T-1}\|
		\sum_{s=0}^{t}M_{s}(W_{eff} - I)^{m}W_{eff}^{t-s}\|_{F}^{2}\\
		&\leq (1-\lambda_{n})^{2m-2}\sum_{t=0}^{T-1}\|M_{t}\|_{F}^{2}.
	\end{align*}
\end{lemma}

\begin{lemma}\label{le:matrix2} For any matrix sequence ${M_{t}}$, we have
	$$
	\sum_{t=0}^{T-1}\|
	\sum_{s=0}^{t}M_{s}(I-A_{n})^{m}W_{eff}^{t-s}\|_{F}^{2}
	\leq \displaystyle\frac{1}{(1-\lambda_{2})^{2}}\sum_{t=0}^{T-1}\|M_{t}\|_{F}^{2}.
	$$
\end{lemma}

	In the above {Lemma \ref{le:matrix1}} and {Lemma \ref{le:matrix2}}, $W$ is a weight matrix and $W_{eff}=(1-\eta)I+\eta W$. $\lambda_{n}=\lambda_{n}(W_{eff})$ and $I$ is an identity matrix.
	\begin{lemma}\label{le:sum f}
		Under \textbf{Assumption~\ref{assumptions}}, setting $\gamma L \leq 1$, we have
		\begin{equation*}
			\begin{aligned}
				\frac{\gamma}{2}&\sum_{t=0}^{T-1}\mathbb{E}\left\|\nabla f(\overline{x}_{t-\frac{1}{2}})\right\|^{2} \leq \mathbb{E}f(\overline{x}_{-\frac{1}{2}}) - \mathbb{E}f({x}^{*})\\
				&+\frac{\gamma L^{2}}{2N}\sum_{t=0}^{T-1}\mathbb{E}\left\|X_{t-\frac{1}{2}}\left(I - \left(\frac{1}{N}\right)_N\right)\right\|_{F}^{2} + \frac{LT\gamma^{2}\sigma_{f}^{2}}{2N},
			\end{aligned}
		\end{equation*}
		where $x^*=\overline{x}_{T-\frac{1}{2}}$ .
	\end{lemma}
	
	\begin{proof}
		By definition, we have $\mathbb{E}(\overline{G}_{t})=\nabla\overline{f}(X_{t-\frac{1}{2}}).$ Then,
		\begin{equation*}\label{eq:E G f}
			\begin{aligned}
				&~~~~\mathbb{E}\left\|\overline{G}_{t}-\nabla\overline{f}(X_{t-\frac{1}{2}})\right\|^2\\
				&=\mathbb{E}\left(\left\|\overline{G}_{t})\right\|^2+\left\|\nabla\overline{f}(X_{t-\frac{1}{2}})\right\|^2-2\left\langle\overline{G}_{t}, \nabla\overline{f}(X_{t-\frac{1}{2}})\right\rangle\right)\\
				&=\mathbb{E}\left\|\overline{G}_{t}\right\|^2-\left\|\nabla\overline{f}(X_{t-\frac{1}{2}})\right\|^2
			\end{aligned}
		\end{equation*}
		According to the updating rules of Algorithm~\ref{Algo1} and Lemma~\ref{le:vector operations}, we have:
		
		\begin{align}\label{eq:gap}
			&\ \ \ \mathbb{E}f(\overline{x}_{t+\frac{1}{2}}) - \mathbb{E}f(\overline{x}_{t-\frac{1}{2}})\nonumber\\
			&\leq -\gamma \mathbb{E}\left\langle \overline{G}_{t}, \nabla f(\overline{x}_{t-\frac{1}{2}})\right\rangle
			+ \displaystyle\frac{L\gamma^{2}}{2}\mathbb{E}\left\|\overline{G}_{t}\right\|^{2}\nonumber\\
			&= -\gamma\mathbb{E}\left \langle \nabla \overline{f}(X_{t-\frac{1}{2}}),\nabla f(\overline{x}_{t-\frac{1}{2}})\right \rangle \nonumber\\
			&~~~~+ \displaystyle\frac{L\gamma^{2}}{2}\mathbb{E}\left\|\overline{G}_{t} - \nabla \overline{f}(X_{t-\frac{1}{2}})\right\|^{2}
			+\displaystyle\frac{L\gamma^{2}}{2}\mathbb{E}\left\|\nabla \overline{f}(X_{t-\frac{1}{2}})\right\|^{2}\nonumber\\
			&\overset{\text{(a)}}{\leq} -\gamma\mathbb{E}\left\langle \nabla \overline{f}(X_{t-\frac{1}{2}}),\nabla f(\overline{x}_{t-\frac{1}{2}})\right\rangle + \displaystyle\frac{L\gamma^{2}}{2}\mathbb{E}\| \nabla \overline{f}(X_{t-\frac{1}{2}})\|^{2} \nonumber\\ 
			&~~~~+ \displaystyle\frac{L\gamma^{2}\sigma_{f}^{2}}{2N}\nonumber\\
			&= -\displaystyle\frac{\gamma}{2}\bigg(\mathbb{E}\left\|\nabla \overline{f}(X_{t-\frac{1}{2}})\right\|^{2} + \mathbb{E}\left\|\nabla f(\overline{x}_{t-\frac{1}{2}})\right\|^{2} \nonumber\\
			&~~~~- \mathbb{E}\left\|\nabla \overline{f}(X_{t-\frac{1}{2}}) - \nabla f(\overline{x}_{t-\frac{1}{2}})\right\|^{2}\bigg)\nonumber\\
			&\ \ \ \ + \displaystyle\frac{L\gamma^{2}}{2}\mathbb{E}\| \nabla \overline{f}(X_{t-\frac{1}{2}})\|^{2} + \displaystyle\frac{L\gamma^{2}\sigma_{f}^{2}}{2N}.
		\end{align}
		
		The above inequality $(a)$ is  based on Cauchy–Schwarz inequality and Assumption~\ref{assumptions}.
		Then we bound the difference between $\nabla \overline{f}(X_{t-\frac{1}{2}}) $ and $\nabla f(\overline{x}_{t-\frac{1}{2}})$:
		\begin{align}\label{eq:f2}
			&~~~~	\mathbb{E} \left\|\nabla \overline{f}(X_{t-\frac{1}{2}}) - \nabla f(\overline{x}_{t-\frac{1}{2}})\right\|^{2}\nonumber\\
			&= \mathbb{E}\left\|\displaystyle\frac{1}{N}\sum_{i=1}^{N}\left(\nabla f_{i}(x_{i}^{(t-\frac{1}{2})}) - 
			\nabla f_{i}(\overline{x}_{t-\frac{1}{2}})\right)\right \|^{2}\nonumber\\
			&\leq \displaystyle\frac{1}{N^{2}}\mathbb{E}\left\|
			\sum_{i = 1}^{N} \left(\nabla f_{i}(x_{i}^{(t-\frac{1}{2})}) - 
			\nabla f_{i}(\overline{x}_{t-\frac{1}{2}})\right)\right\|^{2}\nonumber\\
			&\overset{\text{(b)}}{\leq} \displaystyle\frac{1}{N}\sum_{i=1}^{N}
			\mathbb{E}\left\|\nabla f_{i}(x_{i}^{(t-\frac{1}{2})}) - 
			\nabla f_{i}(\overline{x}_{t-\frac{1}{2}})\right\|^{2}\nonumber\\
			&\overset{\text{(Assumption~\ref{assumptions})}}\leq \displaystyle\frac{L^{2}}{N}\sum_{i=1}^{N}
			\mathbb{E}\left\|x_{i}^{(t-\frac{1}{2})} - \overline{x}_{t-\frac{1}{2}}\right\|^{2}\nonumber\\
			& \leq \displaystyle\frac{L^{2}}{N}\mathbb{E}\left\|
			X_{t-\frac{1}{2}}\left(I - \left(\frac{1}{N}\right)_N\right)\right\|_{F}^{2}.
		\end{align}
		
		The above inequality $(b)$ is based on Cauchy–Schwarz inequality.
		Combining Eqt.~(\ref{eq:gap}) and Eqt.~(\ref{eq:f2}), it can be obtained that
		\begin{equation*}\nonumber
			\begin{aligned}
				\mathbb{E}f(\overline{x}_{t+\frac{1}{2}}) - \mathbb{E}f(\overline{x}_{t-\frac{1}{2}})&\leq- \displaystyle\frac{\gamma}{2} 
				\mathbb{E}\left\|
				\nabla f(\overline{x}_{t-\frac{1}{2}}) \right\|^{2} 				+ \displaystyle\frac{L\gamma^{2}\sigma_{f}^{2}}{2N}\\&-
				(\displaystyle\frac{\gamma}{2} - \displaystyle\frac{L\gamma^{2}}{2})
				\mathbb{E}\left\|
				\nabla \overline{f}(X_{t-\frac{1}{2}})\right\|^{2}\\
				&+ \displaystyle\frac{\gamma L^{2}}{2N}
				\mathbb{E}\left\|
				X_{t-\frac{1}{2}}\left(I - \left(\frac{1}{N}\right)_N\right)\right\|_{F}^{2}.
			\end{aligned}
		\end{equation*}
		
		By summing up the equation above from $t = 0$ to $t = T - 1$ we get
		\begin{equation*}
			\begin{aligned}
				\frac{\gamma}{2}&\sum_{t=0}^{T-1}\mathbb{E}\left\|\nabla f(\overline{x}_{t-\frac{1}{2}})\right\|^{2} \leq \mathbb{E}f(\overline{x}_{-\frac{1}{2}}) - \mathbb{E}f({x}^{*})\\
				&+\frac{\gamma L^{2}}{2N}\sum_{t=0}^{T-1}\mathbb{E}\left\|X_{t-\frac{1}{2}}\left(I - \left(\frac{1}{N}\right)_N\right)\right\|_{F}^{2} + \frac{LT\gamma^{2}\sigma_{f}^{2}}{2N}.
			\end{aligned}
		\end{equation*}
	\end{proof}
	
	Next, we try to upper bound $\sum_{t=0}^{T-1} \mathbb{E} \left\|X_{t-\frac{1}{2}} \left(I - \left(\frac{1}{N}\right)_{N}\right)\right\|_F^2$.
	\begin{lemma}\label{le: for G}
		Under \textbf{Assumption~\ref{assumptions}}, we have
		\begin{equation*}
			\begin{aligned}
				&\sum_{t=0}^{T-1} \mathbb{E} \left\|X_{t-\frac{1}{2}} \left(I - \left(\frac{1}{N}\right)_{N}\right)\right\|_F^2 \leq\\
				& \frac{2\gamma^2 (N-1)^2}{N^2}\sum_{t=0}^{T-1}\mathbb{E} \left\|G_{t}\right \|_{F}^{2} + 2\sigma_{z}^2dNT,
			\end{aligned}
		\end{equation*}
		where $\sigma_z^2 = \max \left\{\frac{|h_k|^2\beta_{k}P_K\sigma^2}{c^2}+ \frac{\sigma_m^2}{c^2(N-1)^2}, k\in V \right \} $.
	\end{lemma}
	
	\begin{proof}
		According to Eqt.~(\ref{equ:X}), we can get
		\begin{equation*}
			\begin{aligned}	
				&~~~~X_{t-\frac{1}{2}}\left(I - \left(\frac{1}{N}\right)_{N}\right) \\
				&= -\gamma \sum_{s=0}^{t-1}G_{s}\Psi^{t-s}\left(I - \left(\frac{1}{N}\right)_{N}\right) \\
				&~~~~+ \sum_{s=0}^{t-1}\Phi_{s}\Psi^{t-s-1}(\Psi-I)\left(I - \left(\frac{1}{N}\right)_{N}\right)\\
				&=-\gamma \sum_{s=0}^{t-1}G_{s}\Psi^{t-s}\left(I - \left(\frac{1}{N}\right)_{N}\right) \\
				&~~~~+ \sum_{s=0}^{t-1}\Phi_{s}\Psi^{t-s-1}(\Psi-I)\\
				&~~~~- \sum_{s=0}^{t-1}\Phi_{s}\Psi^{t-s-1}(\Psi-I)\left(\frac{1}{N}\right)_{N}\\
				&=-\gamma \sum_{s=0}^{t-1}G_{s}\Psi^{t-s}\left(I - \left(\frac{1}{N}\right)_{N}\right) \\
				&~~~~+ \sum_{s=0}^{t-1}\Phi_{s}\Psi^{t-s-1}(\Psi-I)\\
				&~~~~- \sum_{s=0}^{t-1}\eta\Phi_{s}\Psi^{t-s-1}(W-I)\left(\frac{1}{N}\right)_{N}\\
				&=-\gamma \sum_{s=0}^{t-1}G_{s}\Psi^{t-s}\left(I - \left(\frac{1}{N}\right)_{N}\right) \\
				&~~~~+ \sum_{s=0}^{t-1}\Phi_{s}\Psi^{t-s-1}(\Psi-I)\\
				&~~~~-\sum_{s=0}^{t-1}\eta\Phi_{s}\Psi^{t-s-1}\left[W\left(\frac{1}{N}\right)_{N}-\left(\frac{1}{N}\right)_N\right].\\
			\end{aligned}
		\end{equation*}
		
		As $W = \frac{(1)_{N}- I}{N-1}$, the matrix $W$ is doubly stochastic. Then,
		\begin{equation*}
			\begin{aligned}
				X_{t-\frac{1}{2}}\left(I - \left(\frac{1}{N}\right)_{N}\right) &=-\gamma \sum_{s=0}^{t-1}G_{s}\Psi^{t-s}\left(I - \left(\frac{1}{N}\right)_{N}\right) \\
				&~~~~+ \sum_{s=0}^{t-1}\Phi_{s}\Psi^{t-s-1}(\Psi-I).
			\end{aligned}
		\end{equation*}
		
		Hence,according to Cauchy–Schwarz inequality,
		
		\begin{align}\label{eq:X_t}
			&~~~~	\left	\|X_{t-\frac{1}{2}}\left(I - \left(\frac{1}{N}\right)_{N}\right) \right\|_{F}^{2}\nonumber \\
			&\leq 2\gamma^{2}\left\|
			\sum_{s = 0}^{t-1}G_{s}\Psi^{t-s}\left(I - \left(\frac{1}{N}\right)_{N}\right) \right\|_{F}^{2}\nonumber \\
			&~~~~+2\left\|\sum_{s=0}^{t-1}\Phi_{s}(\Psi - I)\Psi^{t-s-1}\right\|_{F}^{2}
		\end{align}
		
		According to Eqt.~(\ref{eq:X_t}), we have
		\begin{equation*}
			\begin{aligned}
				&\ \ \ \sum_{t=0}^{T-1}\mathbb{E}
				\left\|X_{t-\frac{1}{2}}\left(I - \left(\frac{1}{N}\right)_N\right)\right\|_{F}^{2}\\
				&\leq 2\gamma^{2}\sum_{t=0}^{T-1}\mathbb{E}\left\|
				\sum_{s = 0}^{t-1}G_{s}\Psi^{t-s}\left(I-\left(\frac{1}{N}\right)_N\right)\right\|_{F}^{2}\\
				&\ \ +2 \sum_{t=0}^{T-1}\mathbb{E}
				\left\|\sum_{s=0}^{t-1}\Phi_{s}(\Psi - I)\Psi^{t-s-1}\right\|_{F}^{2}\\
				&\leq 2\gamma^{2}\sum_{t=0}^{T-1}\mathbb{E}\left\|
				\sum_{s = 0}^{t}G_{s}\Psi^{t-s}\left(I-\left(\frac{1}{N}\right)_N\right)\right\|_{F}^{2}\\
				&\ \ +2 \sum_{t=0}^{T-2}\mathbb{E}
				\left\|\sum_{s=0}^{t}\Phi_{s}(\Psi - I)\Psi^{t-s}\right\|_{F}^{2}.\\
			\end{aligned}
		\end{equation*}
		
		By Lemma~\ref{le:matrix1} and Lemma~\ref{le:matrix2}, it can be obtained that
		\begin{equation*}
			\begin{aligned}
				&~~~~\sum_{t=0}^{T-1} \mathbb{E} \left\|X_{t-\frac{1}{2}} \left(I - \left(\frac{1}{N}\right)_{N}\right)\right\|_F^2\\
				&\leq \displaystyle\frac{2\gamma^{2}}{(1-\lambda_{2})^{2}}
				\sum_{t=0}^{T-1}\mathbb{E}\left\|G_{t}\right\|_{F}^{2} + 
				2\sum_{t=0}^{T-2}\mathbb{E}\left\|\Phi_{t}\right\|_{F}^{2}\\
				&\leq \frac{2\gamma^2 (N-1)^2}{N^2}\sum_{t=0}^{T-1}\mathbb{E}\left \|G_{t} \right\|_{F}^{2} + 2\sigma_{z}^2dNT.
			\end{aligned}
		\end{equation*}
	\end{proof}
	
	Next, we try to upper bound $\mathbb{E} \left\|G_{t} \right\|_{F}^{2}$.
	
	\begin{lemma}\label{le:G}
		Under \textbf{Assumption~\ref{assumptions}}, we have
		\begin{equation*}	
			\begin{aligned}
				\mathbb{E}\left\|
				G_{t}\right\|_{F}^{2} \leq
				&N\sigma_{f}^{2} + 3L^{2}\mathbb{E} \left\|X_{t-\frac{1}{2}} \left(I - \left(\frac{1}{N}\right)_{N}\right)\right\|_F^2+ 3N\zeta^{2}\\
				&+ 3N\mathbb{E}\left\|\nabla f(\overline{x}_{{t-\frac{1}{2}}})\right\|^{2}
			\end{aligned}
		\end{equation*}
	\end{lemma}
	
	\begin{proof}
		\begin{equation*}
			\begin{aligned}
				&\ \ \ \ \ \mathbb{E}\left\|
				\nabla F_{i}(x_{i}^{(t-\frac{1}{2})},\xi_{i}^{(t)})\right\|^{2}\\
				&= \mathbb{E}\left\|\nabla F_{i}(x_{i}^{(t-\frac{1}{2})}) - \nabla f_{i}(x_{i}^{(t-\frac{1}{2})}) +
				\nabla f_{i}(x_{i}^{(t-\frac{1}{2})})\right\|^{2}\\
				&= \mathbb{E}\|\nabla F_{i}(x_{i}^{(t-\frac{1}{2})}) - \nabla f_{i}(x_{i}^{(t-\frac{1}{2})})\|^{2}
				+\mathbb{E}\| \nabla f_{i}(x_{i}^{(t-\frac{1}{2})})\|^{2}\\
				&{~~~~+2\mathbb{E}\left\langle  \nabla F_{i}(x_{i}^{(t-\frac{1}{2})})- \nabla f_{i}(x_{i}^{(t-\frac{1}{2})}), \nabla f_{i}(x_{i}^{(t-\frac{1}{2})})\right\rangle}\\
				&= \mathbb{E}\left\|\nabla F_{i}(x_{i}^{(t-\frac{1}{2})}) - \nabla f_{i}(x_{i}^{(t-\frac{1}{2})})\right\|^{2}
				+\mathbb{E}\left\| \nabla f_{i}(x_{i}^{(t-\frac{1}{2})})\right\|^{2}\\
				&\leq \sigma_{f}^{2} + \mathbb{E}\Big\|
				(\nabla f_{i}(x_{i}^{(t-\frac{1}{2})}) - \nabla f_{i}(\overline{x}_{t-\frac{1}{2}})) + \nabla f(\overline{x}_{t-\frac{1}{2}}) \\
				&~~~~+ (	\nabla f_{i}(\overline{x}_{t-\frac{1}{2}}) - \nabla f(\overline{x}_{t-\frac{1}{2}})
				)\Big\|^{2}\\
				&\leq \sigma_{f}^{2} +
				3\mathbb{E}\left\|\nabla f_{i}(x_{i}^{(t-\frac{1}{2})}) - \nabla f_{i}(\overline{x}_{t-\frac{1}{2}})\right\|^{2}\\
				&+ 3\mathbb{E}\left\|\nabla f_{i}(\overline{x}_{t-\frac{1}{2}}) - \nabla f(\overline{x}_{t-\frac{1}{2}})\right\|^{2}
				+ 3\mathbb{E}\left\|\nabla f(\overline{x}_{t-\frac{1}{2}})\right\|^{2}\\
				&\leq \sigma_{f}^{2} + 3L^{2}\mathbb{E}\left\| \overline{x}_{t-\frac{1}{2}} - x_{i}^{(t-\frac{1}{2})}\right\|^{2}
				+ 3\zeta^{2}
				+ 3\mathbb{E}\left\| \nabla f(\overline{x}_{t-\frac{1}{2}})\right\|^{2}.
			\end{aligned}
		\end{equation*}
		
		Thus we have
		\begin{equation*}
			\begin{aligned}
				&\ \ \ \ \mathbb{E}\left\|G(X_{t})\right\|_{F}^{2}\\ 
				&= \sum_{i = 1}^{N}\mathbb{E}\left\|
				\nabla F_{i}(x_{i}^{(t-\frac{1}{2})},\xi_{i}^{(t)})\right\|^{2}\\
				&\leq N\sigma_{f}^{2} + 3L^{2}\sum_{i = 1}^{N}\mathbb{E}\left\|\overline{x}_{t-\frac{1}{2}} - x_{i}^{(t-\frac{1}{2})}\right\|^{2}
				+ 3N\zeta^{2}\\
				&~~~~ + 3N\mathbb{E}\left\|
				\nabla f(\overline{x}_{t-\frac{1}{2}})\right\|^{2}\\
				&\leq N\sigma_{f}^{2} + 3L^{2}\mathbb{E} \left\|X_{t-\frac{1}{2}} \left(I - \left(\frac{1}{N}\right)_{N}\right)\right\|_F^2 + 3N\zeta^{2}\\
				&~~~~+ 3N\mathbb{E}\left\|\nabla f(\overline{x}_{{t-\frac{1}{2}}})\right\|^{2}
			\end{aligned}
		\end{equation*}
	\end{proof}
	
	The following theorem describes the main result on the convergence rate of our DWFL algorithm:
	
	\begin{theorem}\label{thm:main}
		Set $L \leq 1$ and $1 - 12L^{2}C_{2}>0$.
		Then by executing DWFL, we have 
		\begin{equation*}
			\begin{aligned}
				&\ \ \ (\frac{\gamma}{2} - \frac{3\gamma^{3}L^{2}C_{2}}{1 - 6C_{2}L^{2}\gamma^{2}})\sum_{t = 0}^{T-1}\mathbb{E}\left\|
				\nabla f(\overline{x}_{{t-\frac{1}{2}}})\right\|^{2}\\
				&\leq \mathbb{E}f(\overline{x}_{-\frac{1}{2}})
				- \mathbb{E}f(x^{*})
				+ \left(\frac{L\gamma^{2}}{2N}+ \displaystyle\frac{C_{2}L^{2}\gamma^{3}}{1 - 6C_{2}L^{2}\gamma^{2}}\right)\sigma_{f}^{2}T\\
				&+\frac{ 3\gamma^{3}L^{2}TC_{2}\zeta}{1 - 6C_{2}L^{2}\gamma^{2}}+\frac{\gamma L^2dT\sigma_{z}^2}{1 - 6C_{2}L^{2}\gamma^{2}},
			\end{aligned}
		\end{equation*}
		where $C_{2} = \left(\displaystyle\frac{N-1}{N}\right)^2$ and $\sigma_z^2 = \max \left\{\frac{|h_k|^2\beta_{k}P_K\sigma^2}{c^2}+ \frac{\sigma_m^2}{c^2(N-1)^2}, k\in V \right \} $.
	\end{theorem}
	\begin{proof}
		According to Lemma~\ref{le:G}, we have
		\begin{equation*}	
			\begin{aligned}
				\mathbb{E}\left\|
				G_{t}\right\|_{F}^{2} \leq
				&N\sigma_{f}^{2} + 3L^{2}\mathbb{E} \left\|X_{t-\frac{1}{2}} \left(I - \left(\frac{1}{N}\right)_{N}\right)\right\|_F^2+ 3N\zeta^{2}\\
				&+ 3N\mathbb{E}\left\|\nabla f(\overline{x}_{{t-\frac{1}{2}}})\right\|^{2}
			\end{aligned}
		\end{equation*}
		
		Then,
		\begin{align}	\label{eq:last1}
			&~~~~\sum_{t=0}^{T-1}	\mathbb{E}\left\|
			G_{t-\frac{1}{2}}\right\|_{F}^{2}\nonumber \\
			&\leq TN\sigma_{f}^{2} + 3L^{2}\sum_{t=0}^{T-1}\mathbb{E} \left\|X_{t-\frac{1}{2}} \left(I - \left(\frac{1}{N}\right)_{N}\right)\right\|_F^2\nonumber \\
			&+ 3TN\zeta^{2}+ 3N\sum_{t=0}^{T-1}\mathbb{E}\left\|\nabla f(\overline{x}_{{t-\frac{1}{2}}})\right\|^{2}
		\end{align}
		
		According to Lemma~\ref{le: for G}, we have
		\begin{align}\label{eq:last2}
			&~~~~\sum_{t=0}^{T-1} \mathbb{E} \left\|X_{t-\frac{1}{2}} \left(I - \left(\frac{1}{N}\right)_{N}\right)\right\|_F^2\nonumber\\
			& \leq  \frac{2\gamma^2 (N-1)^2}{N^2}\sum_{t=0}^{T-1}\mathbb{E} \left\|G_{t}\right \|_{F}^{2} + 2\sigma_{z}^2dNT.
		\end{align}
		
		By Eqt.~(\ref{eq:last1}) and (\ref{eq:last2}), it can be obtained that
		\begin{align}\label{eq:for th1}
			&~~~~(1 - 6L^{2}\gamma^{2}C_{2}) \sum_{t = 0}^{T - 1}\mathbb{E}\left\|
			X_{t-\frac{1}{2}}\left(I -\left(\frac{1}{N}\right)_N \right)\right\|_{F}^{2}\nonumber \\
			&\leq 2\gamma^{2}C_{2}\left(N\sigma_{f}^{2}T + 3N\zeta^{2}T + 3N\sum_{t = 0}^{T - 1}
			\mathbb{E}\left\|\nabla f(\overline{x}_{t-\frac{1}{2}})\right\|^{2} \right)\nonumber \\
			&~~~~+ 2\sigma_{z}^2dNT.
		\end{align}
		
		According to Lemma~\ref{le:sum f}, we have
		\begin{align}\label{eq:for th2}
			&~~~~\frac{\gamma}{2}\sum_{t=0}^{T-1}\mathbb{E}\left\|\nabla f(\overline{x}_{t-\frac{1}{2}})\right\|^{2}\nonumber \\
			& \leq \mathbb{E}f(\overline{x}_{-\frac{1}{2}}) - \mathbb{E}f({x}^{*})+ \frac{LT\gamma^{2}\sigma_{f}^{2}}{2N}\nonumber \\
			&~~~~+\frac{\gamma L^{2}}{2N}\sum_{t=0}^{T-1}\mathbb{E}\left\|X_{t}\left(I - \left(\frac{1}{N}\right)_N\right)\right\|_{F}^{2}.
		\end{align}
		
		Combining Eqt.~(\ref{eq:for th1}) and (\ref{eq:for th2}), it can be obtained that
		\begin{equation*}
			\begin{aligned}
				&\ \ \ (\frac{\gamma}{2} - \frac{3\gamma^{3}L^{2}C_{2}}{1 - 6C_{2}L^{2}\gamma^{2}})\sum_{t = 0}^{T-1}\mathbb{E}\left\|
				\nabla f(\overline{x}_{{t-\frac{1}{2}}})\right\|^{2}\\
				&\leq \mathbb{E}f(\overline{x}_{-\frac{1}{2}})
				- \mathbb{E}f(x^{*})
				+ \left(\frac{L\gamma^{2}}{2N}+ \displaystyle\frac{C_{2}L^{2}\gamma^{3}}{1 - 6C_{2}L^{2}\gamma^{2}}\right)\sigma_{f}^{2}T\\
				&~~~~+\frac{ 3\gamma^{3}L^{2}TC_{2}\zeta}{1 - 6C_{2}L^{2}\gamma^{2}}+\frac{\gamma L^2dT\sigma_{z}^2}{1 - 6C_{2}L^{2}\gamma^{2}}.
			\end{aligned}
		\end{equation*}
	\end{proof}
	If we choose proper step size $\gamma$, then we immediately get Corollary~\ref{cor:main}.
	\begin{corollary}\label{cor:main}
		If we set $\gamma = \displaystyle\frac{1}{\sigma_{f}}\displaystyle\sqrt{\frac{2C_{4}N}{LT}}$ and  $C_{4} = \mathbb{E}f(\overline{x}_{-\frac{1}{2}}) - \mathbb{E}f(x^{*})$, it holds that 
		\begin{equation*}
			\begin{aligned}
				\displaystyle\frac{1}{T}\sum_{t = 0}^{T-1}\mathbb{E}\left\|
				\nabla f(\overline{x}_{t-\frac{1}{2}})\right\|^{2}&\leq 2\sigma_{f}\displaystyle\sqrt{\frac{2C_{4}L}{TN}}
				+ \displaystyle\frac{20C_2C_4NLd}{T \sigma_{f}^2 \epsilon^2}\\
				&= \mathcal{O}\left(\displaystyle\sqrt{\frac{1}{TN}} \right),
			\end{aligned}
		\end{equation*}
	where $\epsilon \geq  \frac{2\gamma g_{\max}^{(t)}\sqrt{|h_j|^2P_j}}{\sigma_{z}C(N-1)}\sqrt{2\ln \frac{1.25}{\delta}}$.
	\end{corollary}
	
From the corollary above we can see that when $T$ goes to infinity, the dominate term is $\sqrt{\frac{1}{TN}}$ that does not include the privacy budget, which is a very desirable property showing DWFL has great advantages in noise resistance.

\section{Simulation Results}\label{sec:exp}
In this section, we evaluate the performance of our wireless decentralized learning strategy by extensive experiments. Specifically, we first test the impact of different transmitting power to the convergence rate. In addition, we test the convergence rate when the number of workers changes. We also analyze the convergence performance with different $\epsilon$. Then, we show the different converging results in orthogonal setting and our non-orthogonal setting. We also evaluate the different convergence rate between the centralized topology~\cite{20DBLP:conf/isit/SeifTL20} and our decentralized topology at the end. We set $P=60$ dBm, $\epsilon=0.5$ and the algorithm works in a non-orthogonal decentralized strategy. Based on this setting, we study the impact of different parameters or strategies.

Our experiments are conducted on a real dataset cifar-10, and we use Pytorch as our distributed framework. We use 4 GPUs of type GTX 1080Ti to conduct our experiments and each GPU is used to simulate decentralized workers.

\begin{figure}[!hbt]
    \centering
    \subfigure[10 workers]{\includegraphics[width=0.47\linewidth]{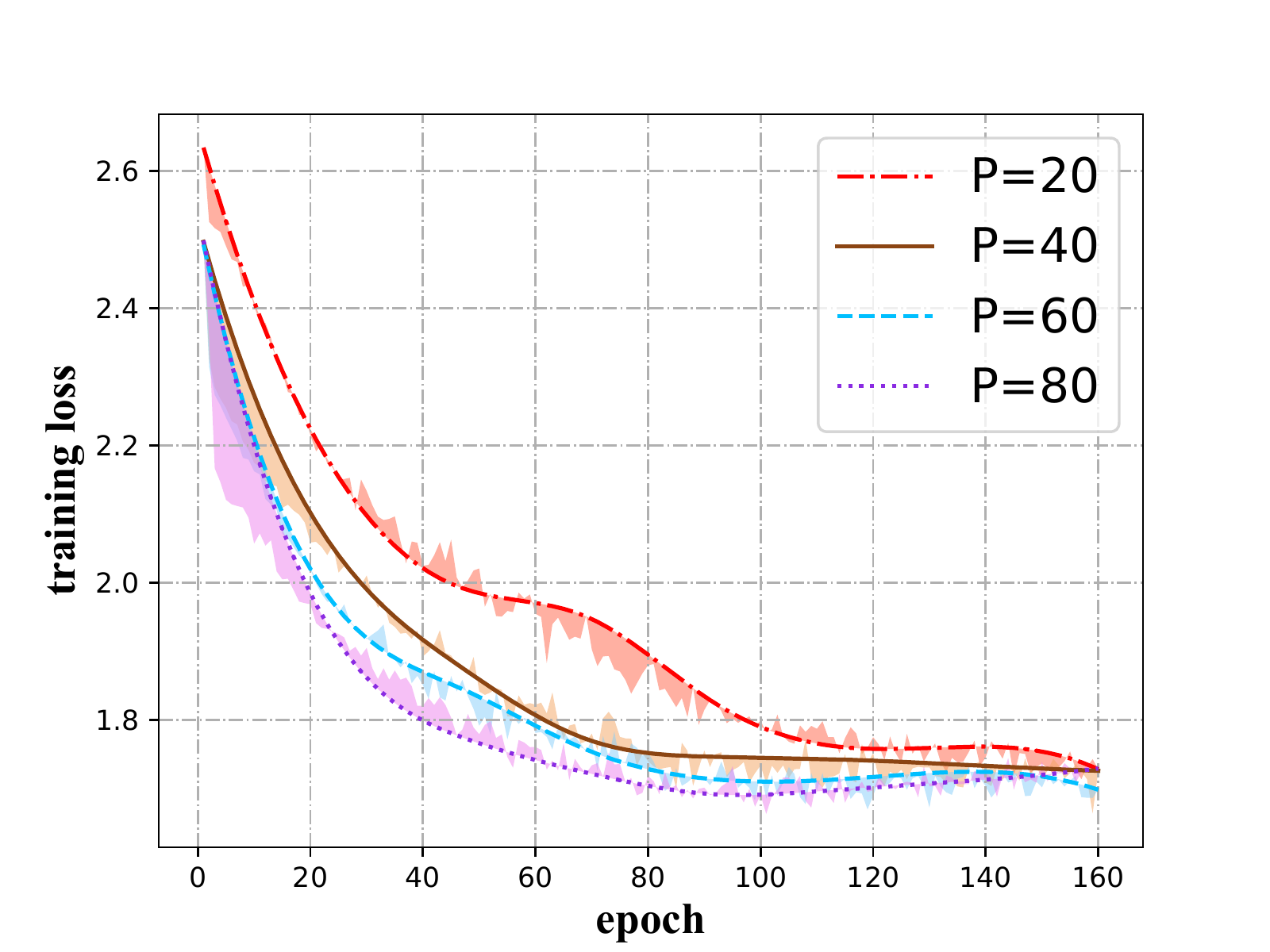}}
    \quad
     \subfigure[30 workers]{\includegraphics[width=0.47\linewidth]{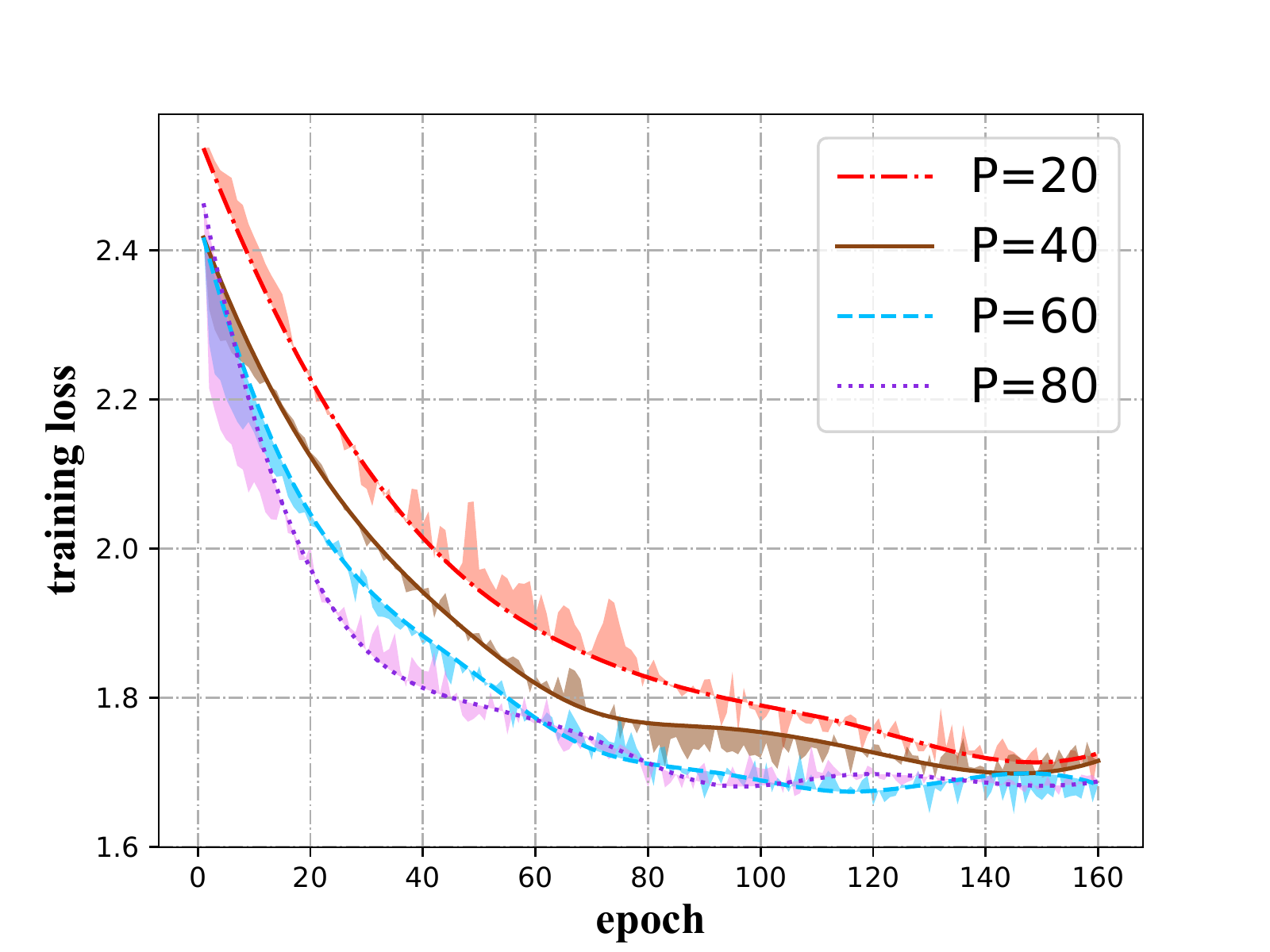}}
     \centering
     \caption{Convergence rate of DWFL when $P$ changes}\label{fig1}
\end{figure}

\begin{figure}[!hbt]
    \centering
    \subfigure[10 workers]{\includegraphics[width=0.47\linewidth]{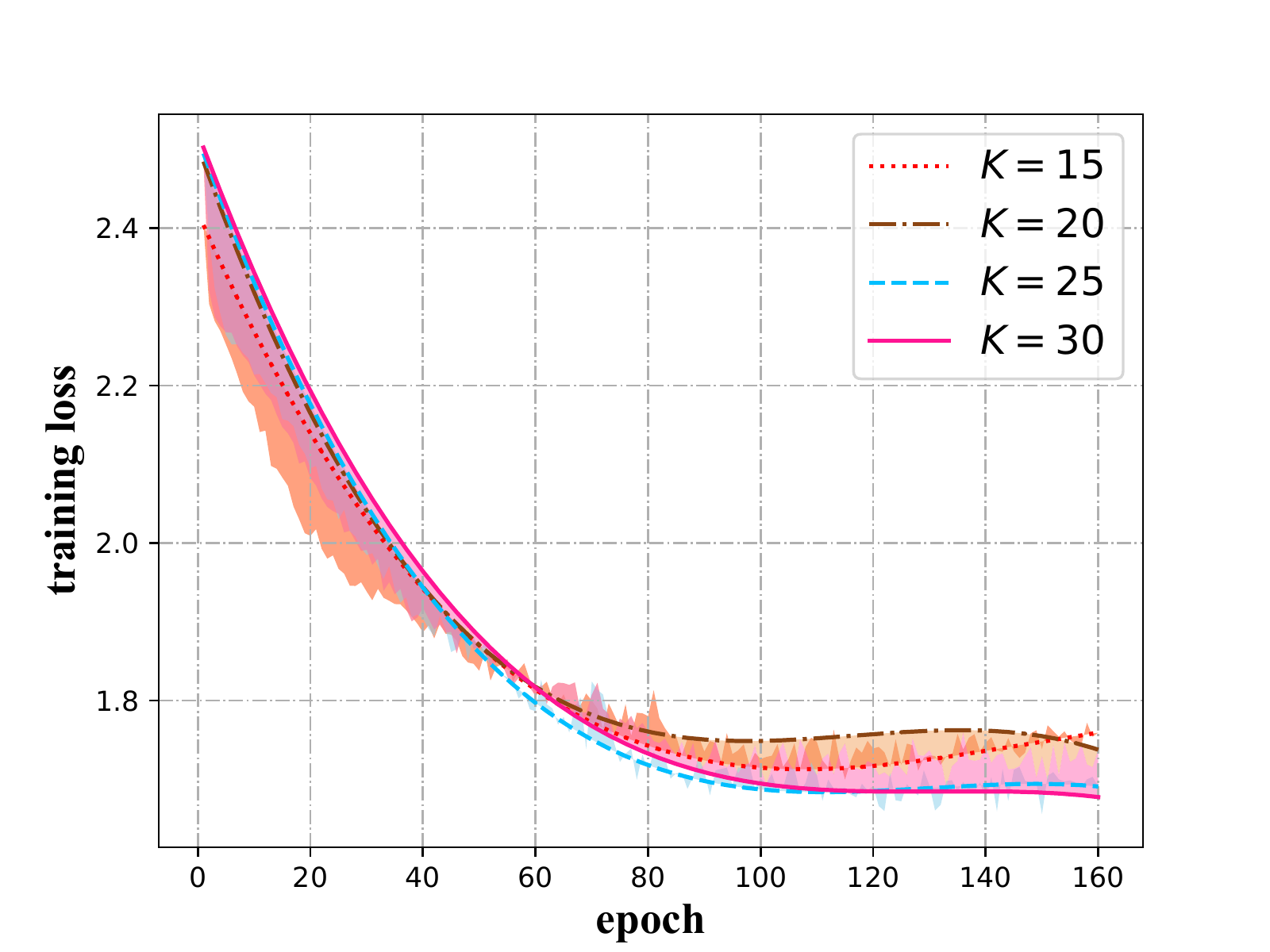}}
    \quad
     \subfigure[30 workers]{\includegraphics[width=0.47\linewidth]{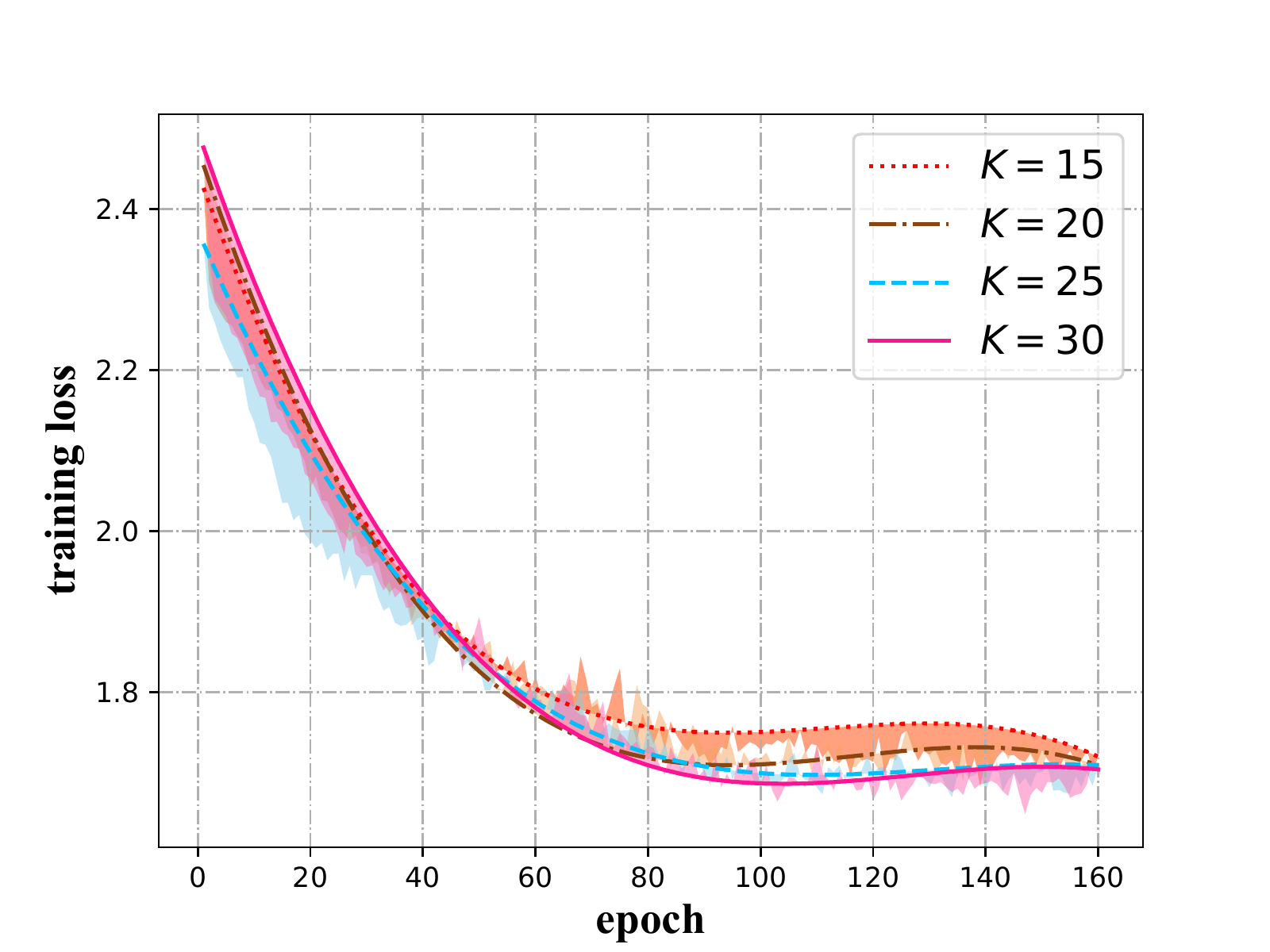}}
     \centering
     \caption{Convergence rate of DWFL when $N$ changes}\label{fig4}
\end{figure}

\begin{figure}[!hbt]
    \centering
    \subfigure[10 workers]{\includegraphics[width=0.47\linewidth]{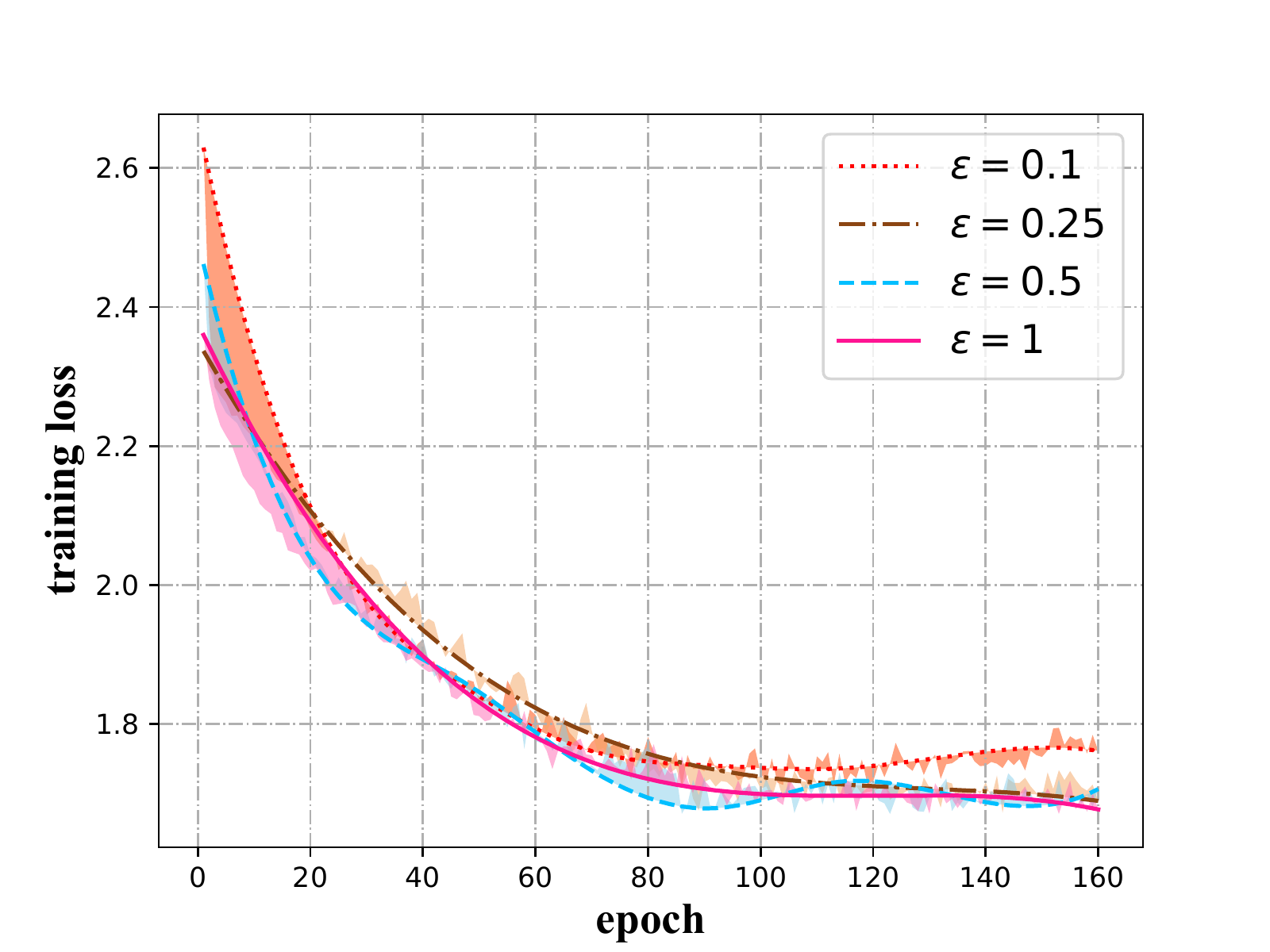}}
    \quad
     \subfigure[30 workers]{\includegraphics[width=0.47\linewidth]{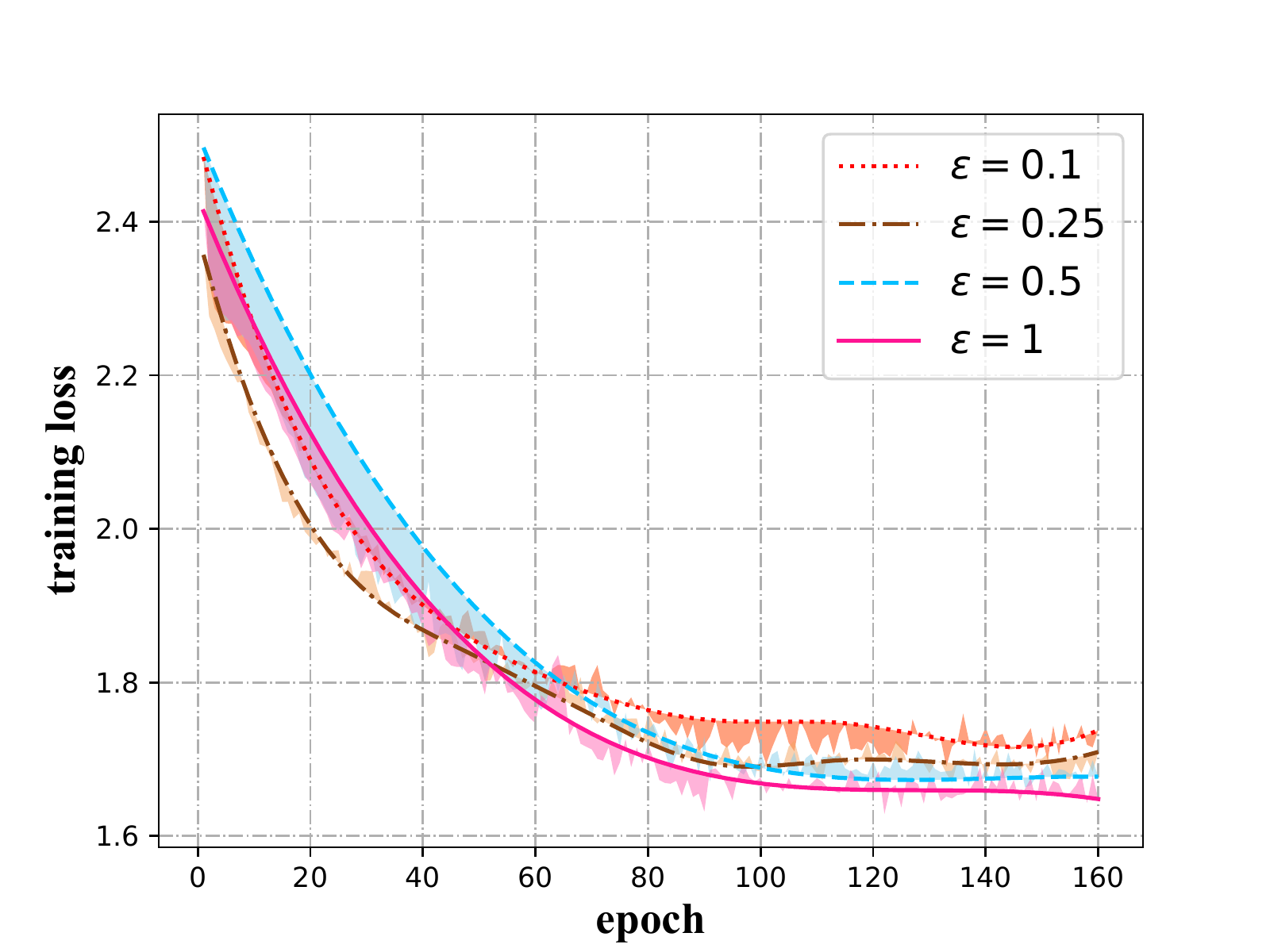}}
     \centering
     \caption{Convergence rate of DWFL when $\epsilon$ changes}\label{fig5}
\end{figure}

\begin{figure}[!hbt]
    \centering
    \subfigure[10 workers]{\includegraphics[width=0.47\linewidth]{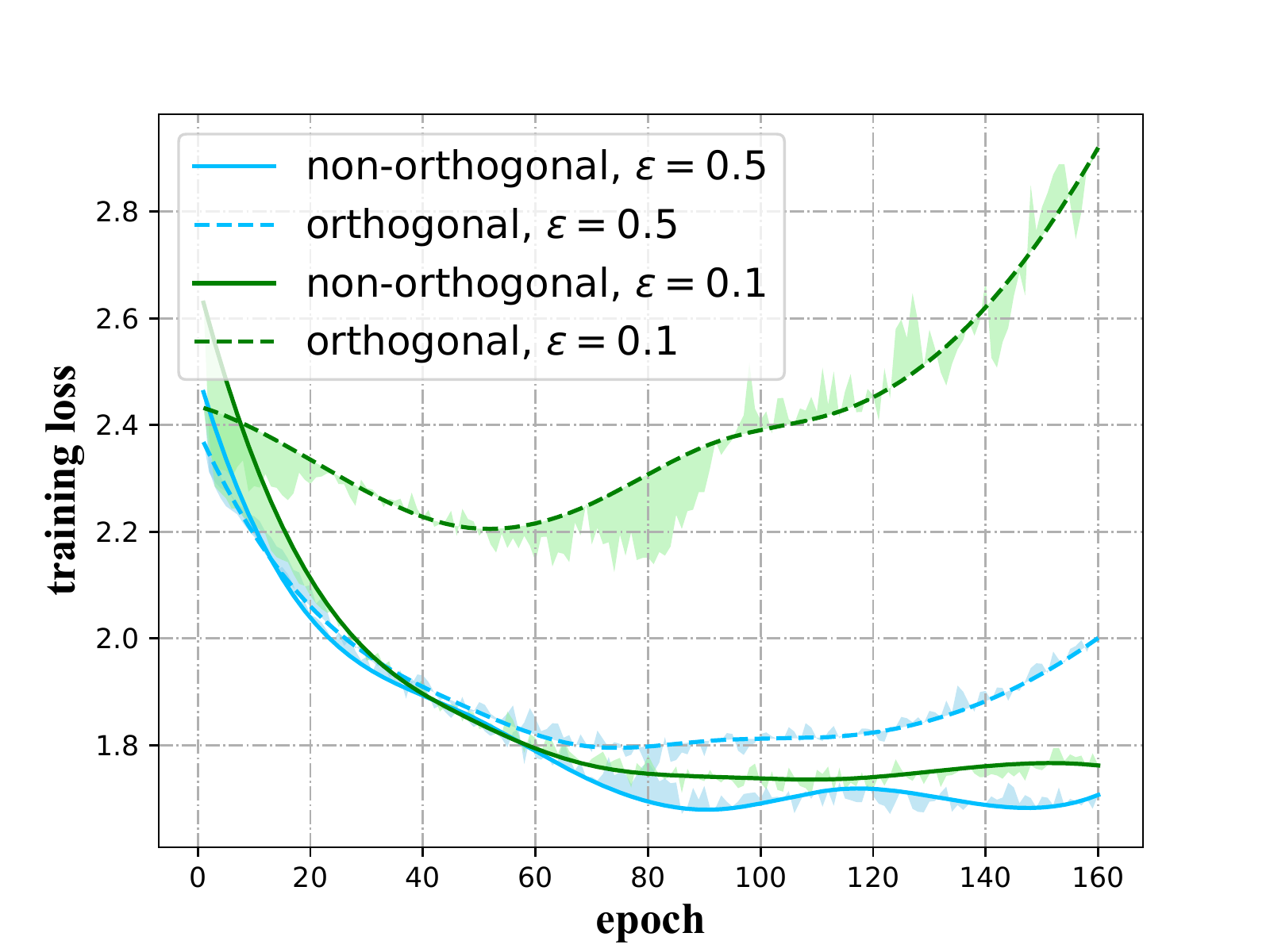}}
    \quad
     \subfigure[30 workers]{\includegraphics[width=0.47\linewidth]{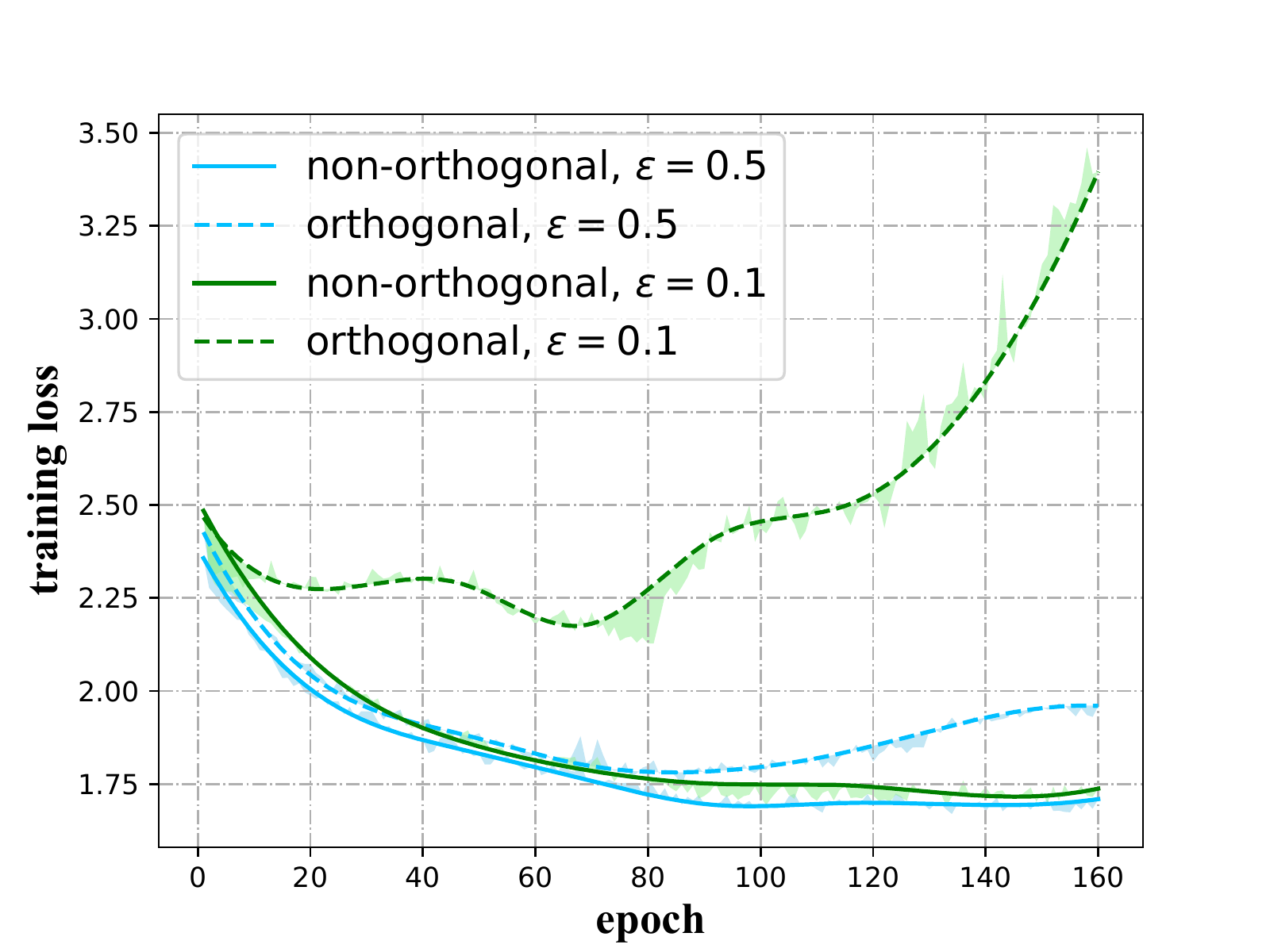}}
     \centering
     \caption{Convergence rate of non-orthogonal strategy and orthogonal strategy.}\label{fig2}
\end{figure}

\begin{figure}[!hbt]
    \centering
    \subfigure[10 workers]{\includegraphics[width=0.47\linewidth]{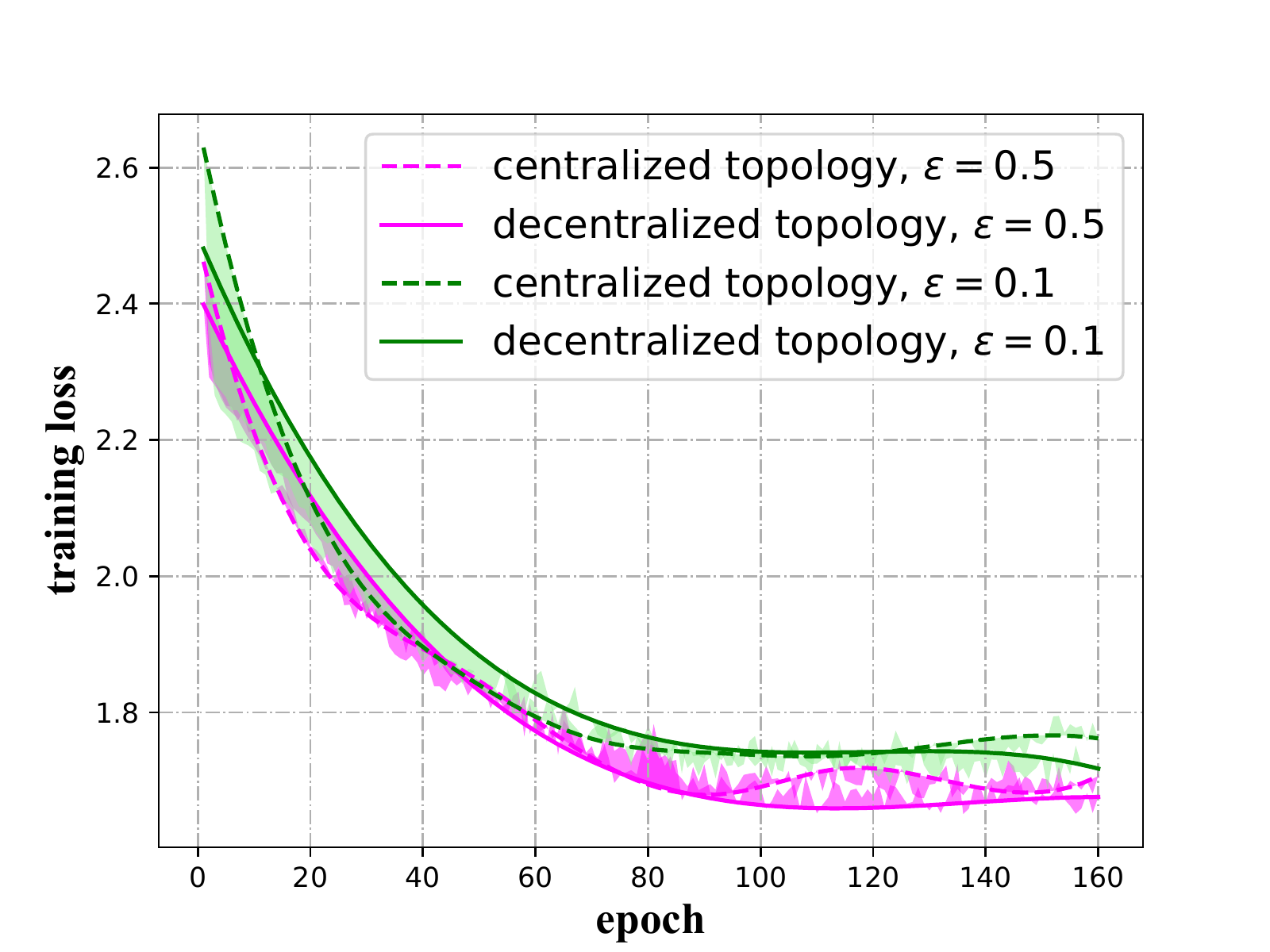}}
    \quad
     \subfigure[30 workers]{\includegraphics[width=0.47\linewidth]{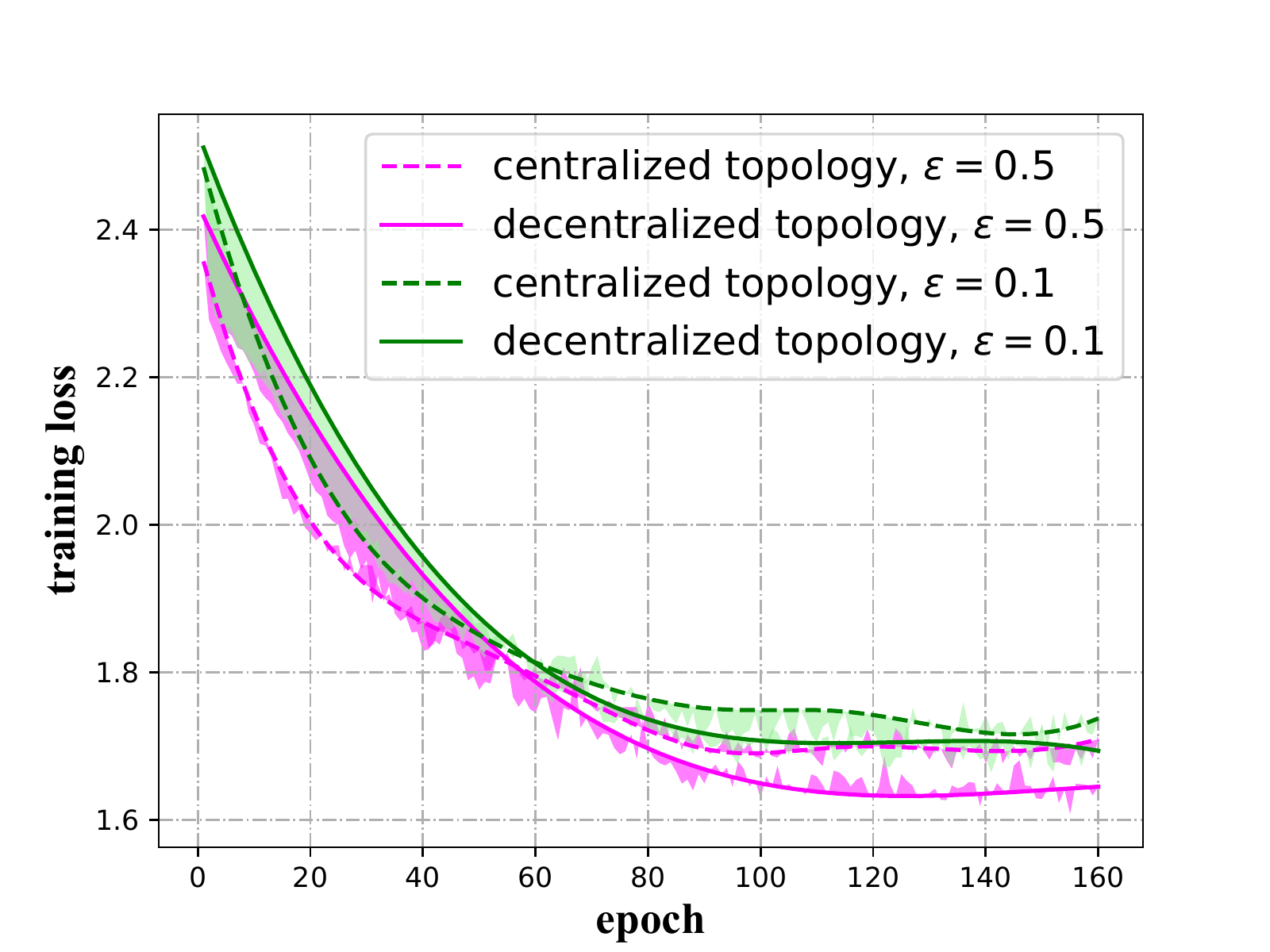}}
     \centering
     \caption{Convergence rate of centralized strategy and decentralized strategy.}\label{fig3}
\end{figure}

In Fig. \ref{fig1}, we show the convergence results with different transmitting powers as $20$ dBm, $40$ dBm, $60$ dBm and $80$ dBm. We set the number of workers $N=10$ and $N=30$ respectively. It shows that stronger transmitting power leads to more rapid convergence, it is reasonable since greater power has stronger ability to resist the interference of channel noise.

In Fig.~\ref{fig4}, we study the performance of our proposed algorithm DWFL with different number of workers from $N=15$, $N=20$, $N=25$ and $N=30$. We set the privacy budget as $\epsilon=0.1$ and $\epsilon=0.5$ respectively. It shows that DWFL performs better with more workers, which is consistent with our analysis.

In Fig.~\ref{fig5}, we show the convergence rate of DWFL with different privacy budget. Specifically, we set $\epsilon=0.1$, $\epsilon=0.25$, $\epsilon=0.5$ and $\epsilon=1$, respectively. It shows that DWFL has better convergence results with smaller privacy budget. Smaller $\epsilon$ means that more noise is added, which dampens the learning process.

In Fig. \ref{fig2}, we compare the convergence rate of orthogonal scheme and our non-orthogonal scheme with $N=10$ and $N=30$ respectively. We show that our non-orthogonal scheme has better advantage over orthogonal scheme in the case of the same privacy level. And our non-orthogonal algorithm converges faster in both $N=10$ and $N=30$. In the meanwhile, the orthogonal scheme almost failed to converge when $\epsilon$ = 0.1 and $N=10$ obviously. This is because when the privacy budget $\epsilon$ becomes smaller, more noise needs to be added, which affects the convergence of the algorithm.

Last but not least, we also evaluate the different convergence rate between the centralized topology and our decentralized topology in Fig. \ref{fig3}. We set the number of workers $N=10$ and $N=30$ respectively. It shows that our decentralized algorithm DWFL acquires more robustness and better convergence results in contrast with that in centralized setting. It is because that the more frequent information exchanges in the single-hop wireless network facilitates the learning process.

\section{Conclusion}\label{sec:con}
In this paper, we studied the decentralized wireless federated learning problem under the
requirement of differential privacy. We proposed an algorithm DWFL which satisfies $\epsilon$-DP. With detailed analysis and sufficient experiments, we showed
that our proposed
algorithms converges at the same rate $\mathcal{O}(\sqrt{\frac{1}{TN}})$ as the centralized algorithm. And our algorithm has great advantages in noise resistance compared with the orthogonal transmission scheme. Based on our work, it is meaningful
to further investigate decentralized wireless algorithms considering reducing communication cost.

%
%
%
%
%

\ifCLASSOPTIONcaptionsoff
  \newpage
\fi



\bibliographystyle{IEEEtran}
\normalem
\bibliography{ref}

\vspace{0cm}
\begin{IEEEbiography}[{\includegraphics[width=1in,height=1.25in,clip,keepaspectratio]{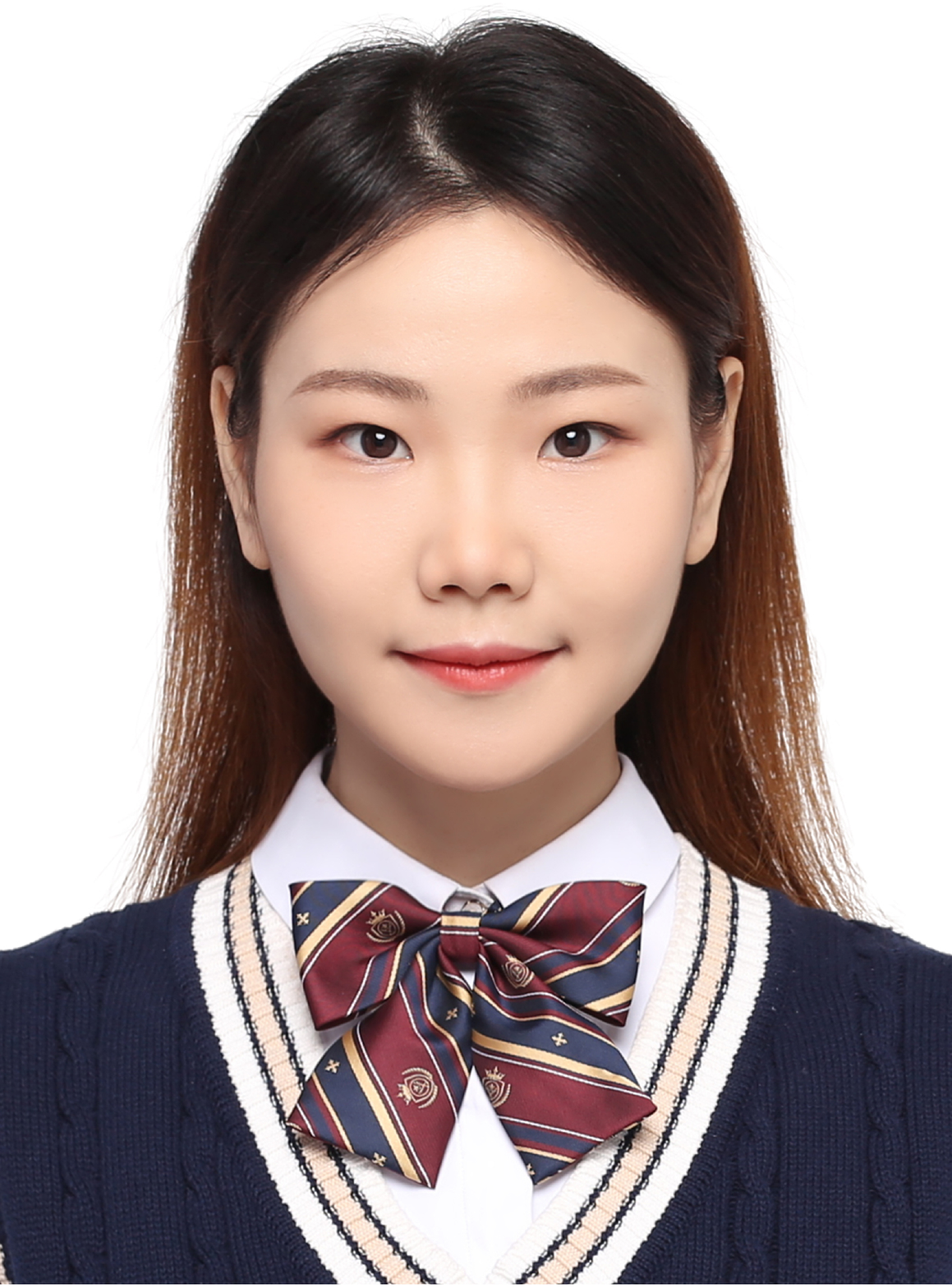}}]{Shuzhen Chen} 
	received the B.S. degree in 2019 from the School of Computer Science and Technology, Shandong University. She is currently pursuing the Ph.D. degree in School of Computer Science and Technology, Shandong University. Her research interests include distributed computing, wireless and mobile security.
\end{IEEEbiography}

\vspace{-1cm}
\begin{IEEEbiography}[{\includegraphics[width=1in,height=1.25in,clip,keepaspectratio]{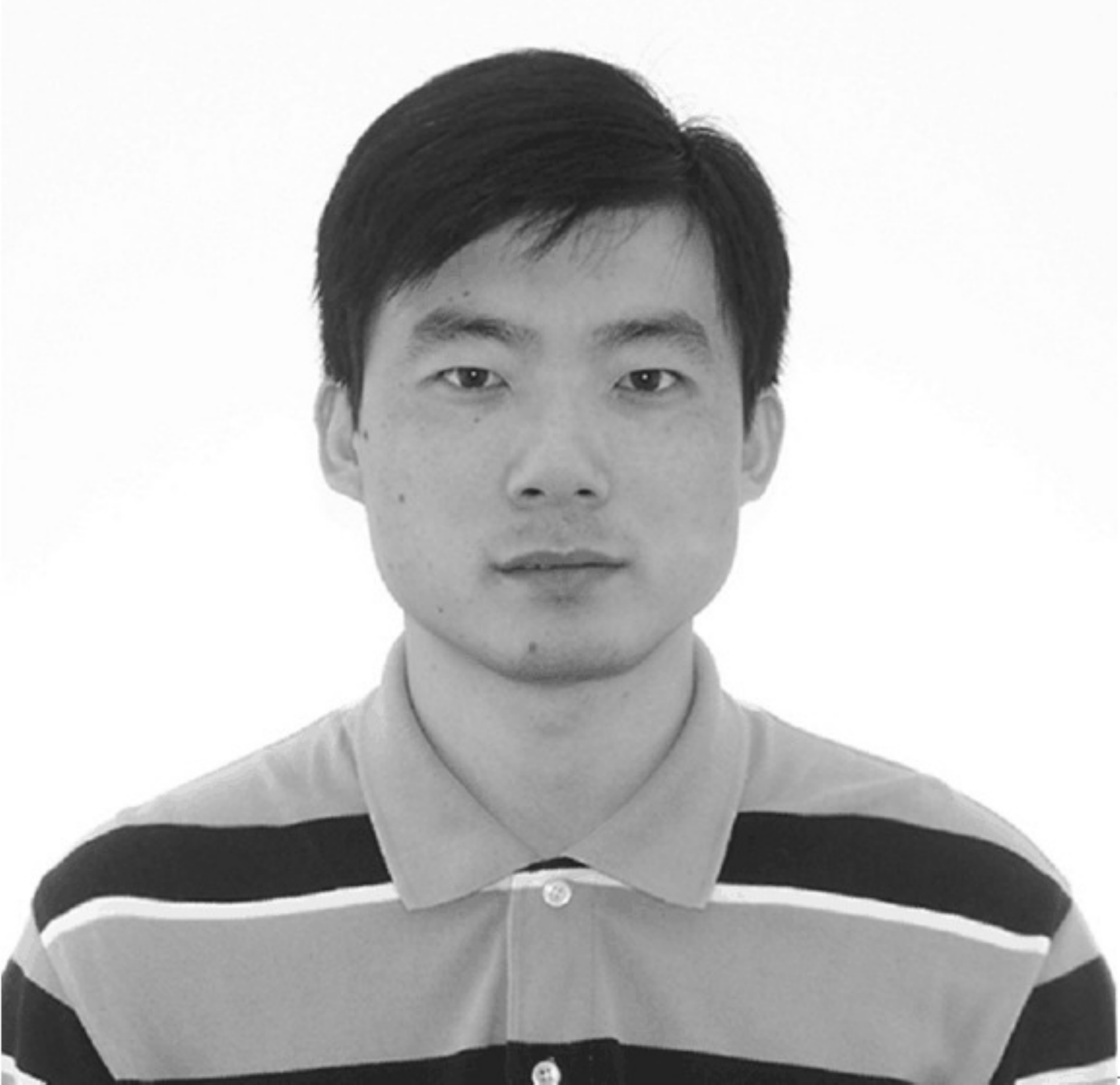}}]{Dongxiao Yu} 
	received the BSc degree in 2006 from the School of Mathematics, Shandong University and the PhD degree in 2014 from the Department of Computer Science, The University of Hong Kong. He became an associate professor in the School of Computer Science and Technology, Huazhong University of Science and Technology, in 2016. He is currently a professor in the School of Computer Science and Technology, Shandong University. His research interests include wireless networks, distributed computing and graph algorithms.
\end{IEEEbiography}

\begin{IEEEbiography}[{\includegraphics[width=1in,height=1.25in,clip,keepaspectratio]{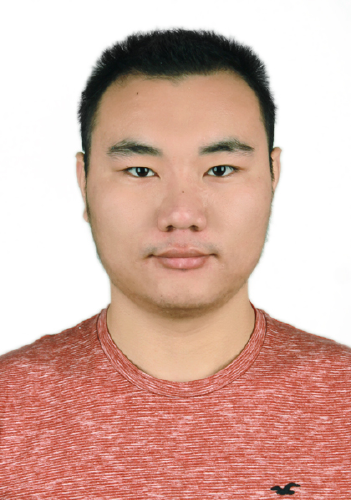}}]{Yifei Zou} 
received the B.E. degree in 2016 from Computer
School, Wuhan University, and the PhD degree in 2020 from
the Department of Computer Science, The University of
Hong Kong. He is currently an Assistant Professor with
the school of computer science and technology, Shandong
University, Qingdao. His research interests include wireless
networks, ad hoc networks and distributed computing.
\end{IEEEbiography}

\vspace{-9cm}
\begin{IEEEbiography}[{\includegraphics[width=1in,height=1.25in,clip,keepaspectratio]{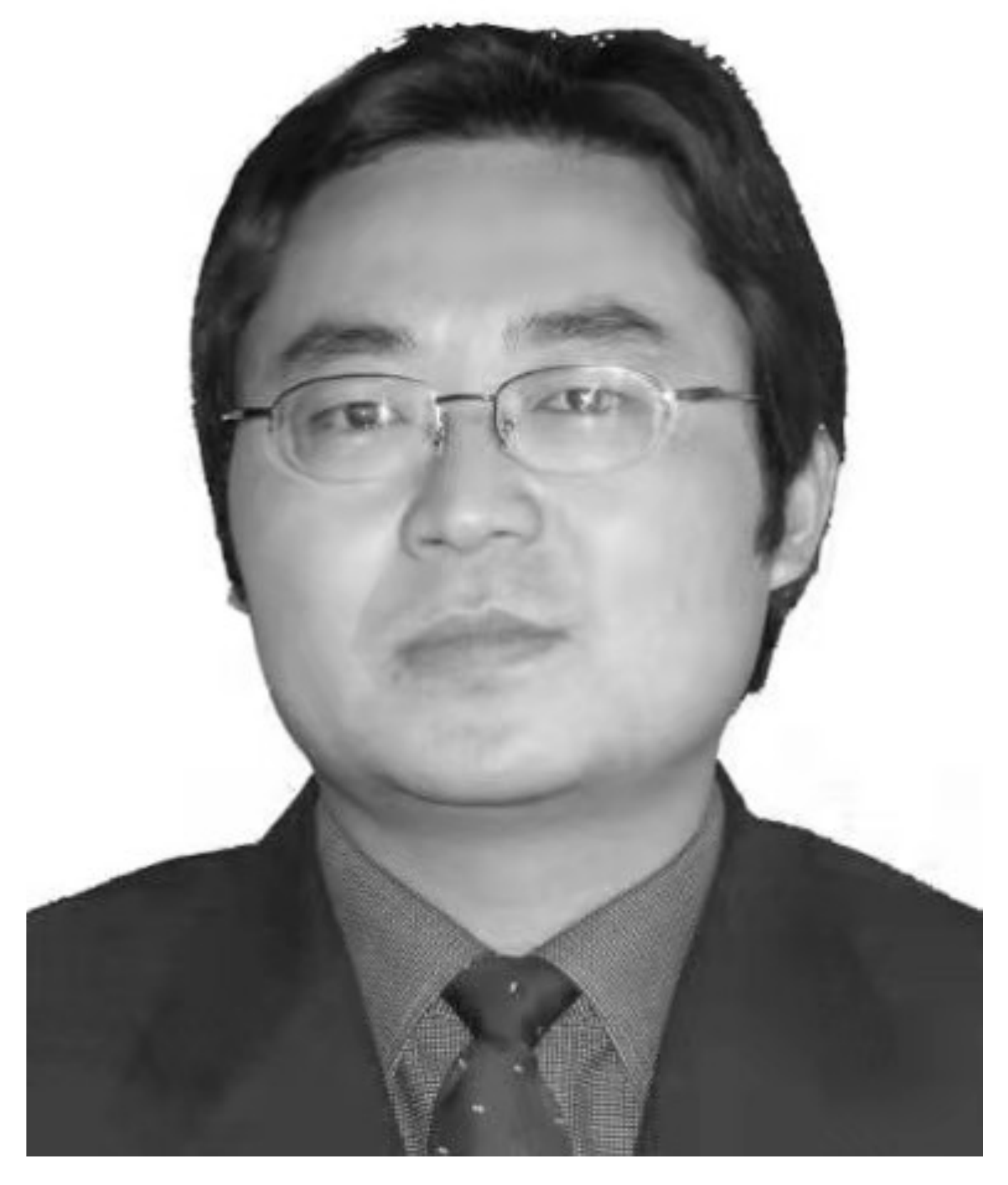}}]{Jiguo Yu} 
	Jiguo Yu (Fellow, IEEE) received the Ph.D. degree from the School of Mathematics, Shandong University, Jinan, China, in 2004. He became a full professor with the School of Computer Science, Qufu Normal University, Jining, China in 2007. He is currently a Full professor with the Qilu University of Technology (Shandong Academy of Sciences), Shandong Computer Science Center (National Supercomputer Center in Jinan), and the Shandong Laboratory of Computer Networks in Jinan, China. His main research interests include privacy-aware computing, wireless networking, distributed algorithms, blockchain, and graph theory.
\end{IEEEbiography}

\vspace{-9cm}
\begin{IEEEbiography}[{\includegraphics[width=1in,height=1.25in,clip,keepaspectratio]{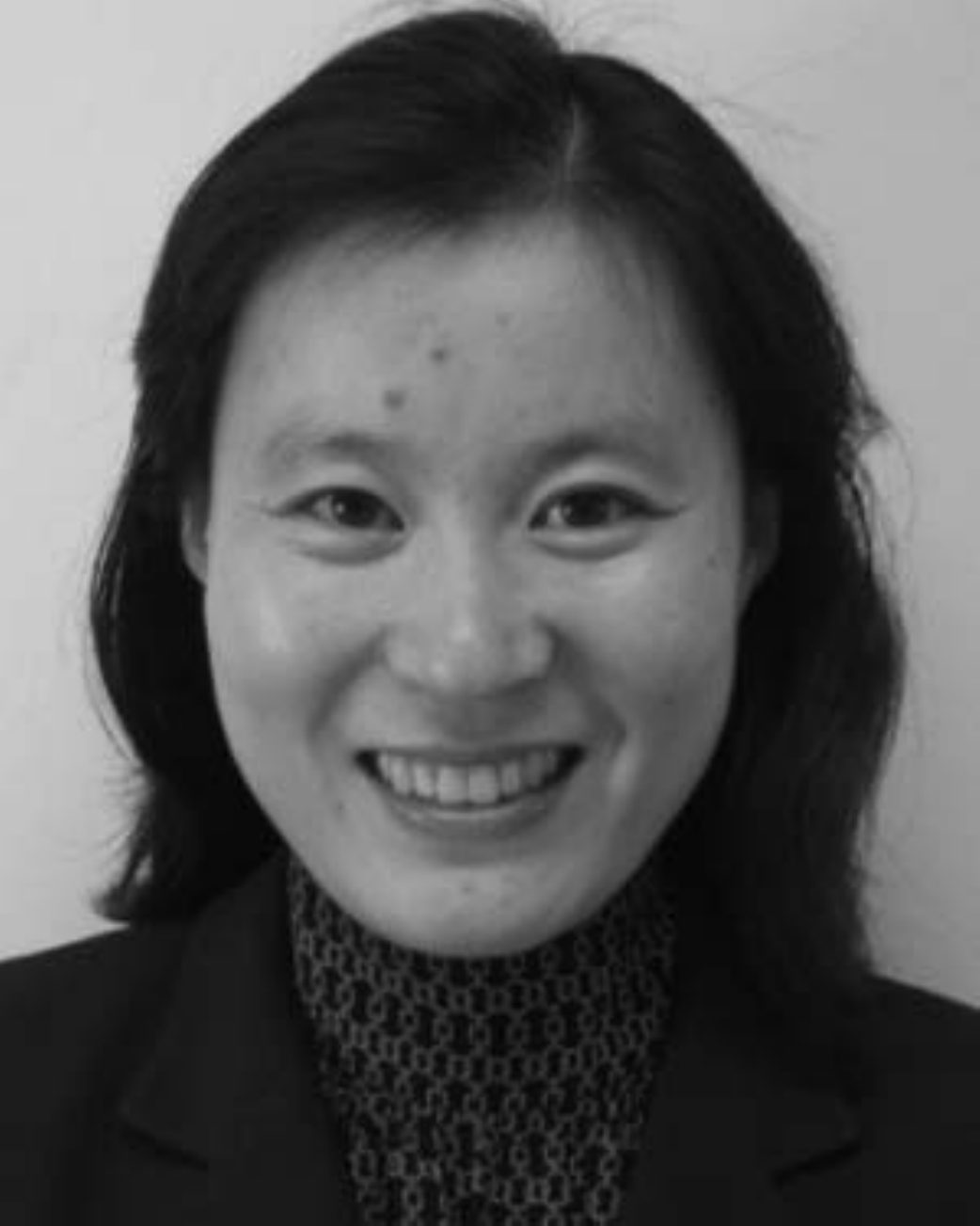}}]{Xiuzhen Cheng} 
	received her M.S. and Ph.D. degrees in computer science from the University of Minnesota -- Twin Cities in 2000 and 2002, respectively. She is a professor in the School of Computer Science and Technology, Shandong University. Her current research interests include cyber physical systems, wireless and mobile computing, sensor networking, wireless and mobile security, and algorithm design and analysis. She has served on the editorial boards of several technical journals and the technical program committees of various professional conferences/workshops. She also has chaired several international conferences. She worked as a program director for the US National Science Foundation (NSF) from April to October in 2006 (full time), and from April 2008 to May 2010 (part time). She received the NSF CAREER Award in 2004. She is Fellow of IEEE and a member of ACM.
\end{IEEEbiography}
\end{document}